\documentclass[reprint,superscriptaddress,noshowpacs,longbibliography,amsmath,amssymb, aps, pre]{revtex4-1}

\usepackage{graphicx}
\graphicspath{{../}}
\usepackage{color}
\definecolor{monbbleu}{RGB}{76, 114, 176}
\usepackage{hyperref}
\hypersetup{
    unicode=false,        
    pdftoolbar=false,     
    pdfmenubar=false,     
    pdffitwindow=false,   
    pdfstartview={FitH},  
    pdftitle={},          
    pdfauthor={},         
    pdfsubject={},        
    pdfkeywords={},       
    pdfnewwindow=true,    
    colorlinks=true,      
    allcolors=monbbleu,   
}
\usepackage{amsmath,amssymb,mathrsfs,amsthm}

\renewcommand{\vec}[1]{\boldsymbol{#1}}

\newcommand{\mean}[1]{\left\langle #1 \right\rangle}

\newcommand{\pder}[2]{\ensuremath\frac{\partial #1}{\partial #2}}
\newcommand{\set}[1]{\left\{#1\right\}}

\newcommand{\bigo}[1]{\ensuremath\mathcal{O}\left(#1\right)}
\newcommand{\mi}{\ensuremath I(X;G)}

\newcommand{\recon}{\ensuremath U(G\,|\,X)}
\newcommand{\pred}{\ensuremath U(X\,|\,G)}
\newcommand{\avgk}{\ensuremath\langle k \rangle}

\theoremstyle{plain}
   \newtheorem{theorem}{Theorem}
   
   \newtheorem{lemma}{Lemma}
   \newtheorem{conjecture}{Conjecture}
   
\theoremstyle{definition}
   \newtheorem{definition}{Definition}
   
\theoremstyle{remark}

\begin{document}
%
%
\title{Duality between predictability and reconstructability in complex systems}
\author{Charles Murphy}%
\email{charles.murphy.1@ulaval.ca}
\affiliation{D\'epartement de physique, de g\'enie physique et d'optique, Universit\'e Laval, Qu\'ebec (Qu\'ebec), Canada G1V 0A6}%
\affiliation{Centre interdisciplinaire en mod\'elisation math\'ematique, Universit\'e Laval, Qu\'ebec (Qu\'ebec), Canada G1V 0A6}%
\author{Vincent Thibeault}%
\affiliation{D\'epartement de physique, de g\'enie physique et d'optique, Universit\'e Laval, Qu\'ebec (Qu\'ebec), Canada G1V 0A6}%
\affiliation{Centre interdisciplinaire en mod\'elisation math\'ematique, Universit\'e Laval, Qu\'ebec (Qu\'ebec), Canada G1V 0A6}%
\author{Antoine Allard}%
\affiliation{D\'epartement de physique, de g\'enie physique et d'optique, Universit\'e Laval, Qu\'ebec (Qu\'ebec), Canada G1V 0A6}%
\affiliation{Centre interdisciplinaire en mod\'elisation math\'ematique, Universit\'e Laval, Qu\'ebec (Qu\'ebec), Canada G1V 0A6}%
\author{Patrick Desrosiers}%
\affiliation{D\'epartement de physique, de g\'enie physique et d'optique, Universit\'e Laval, Qu\'ebec (Qu\'ebec), Canada G1V 0A6}%
\affiliation{D\'epartement de physique, de g\'enie physique et d'optique, Universit\'e Laval, Qu\'ebec (Qu\'ebec), Canada G1V 0A6}%
\affiliation{Centre interdisciplinaire en mod\'elisation math\'ematique, Universit\'e Laval, Qu\'ebec (Qu\'ebec), Canada G1V 0A6}%
\affiliation{Centre de recherche CERVO, Qu\'ebec (Qu\'ebec), Canada G1J 2G3}%

\date{\today}
\begin{abstract}
    Predicting the evolution of a large system of units using its structure of interaction is a fundamental problem in complex system theory. And so is the problem of reconstructing the structure of interaction from temporal observations. Here, we find an intricate relationship between predictability and reconstructability using an information-theoretical point of view. We use the mutual information between a random graph and a stochastic process evolving on this random graph to quantify their codependence. Then, we show how the uncertainty coefficients, which are intimately related to that mutual information, quantify our ability to reconstruct a graph from an observed time series, and our ability to predict the evolution of a process from the structure of its interactions. Interestingly, we find that predictability and reconstructability, even though closely connected by the mutual information, can behave differently, even in a dual manner. We prove how such duality universally emerges when changing the number of steps in the process, and provide numerical evidence of other dualities occurring near the criticality of multiple different processes evolving on different types of structures.
\end{abstract}
\maketitle

\section{Introduction}
\label{sec:introduction}
The relationship between \emph{structure} and \emph{function} is fundamental in complex systems~\cite{barabasi2013network, latora2017complex, newman2018networks}, and important efforts have been invested in developing network models to better understand it.
In particular, models of dynamics on networks~\cite{barzel2013universality, pastor2015epidemic, Boccaletti2016, iacopini2019simplicial} have been proposed to assess the influence of network structure over the temporal evolution of the activity in the system.
In turn, data-driven models~\cite{hebert2020macroscopic, murphy2021deep}, dimension-reduction techniques~\cite{gao2016universal, laurence2019spectral, pietras2019network, thibeault2020threefold} and mean-field frameworks~\cite{pastor2001epidemic, hebert2015complex, st2018phase, st2021master, st2021universal} have deepened our predictive capabilities.
Among other things, these theoretical approaches have shed light on the relationship between dynamics criticality and many network properties such as the degree distribution~\cite{pastor2001epidemic, st2018phase}, the eigenvalue spectrum~\cite{ferreira2012epidemic, castellano2017relating, pastor2018eigenvector} and their group structure~\cite{hebert2010propagation, st2021social, st2021master}.
Fundamentally, these contributions justify our inclination for measuring and using real-world networks as a proxy to predict the behavior of complex systems.

Models of dynamics on networks have also been used as reverse engineering tools for network reconstruction~\cite{brugere2018network}, when the networks of interactions are unavailable, noisy~\cite{peixoto2018reconstructing, young2020bayesian, young2021reconstruction} or faulty~\cite{laurence2020detecting}.
The network reconstruction problem has stimulated many technical contributions~\cite{mccabe2020netrd}: Thresholding matrices built from correlation~\cite{kramer2009network} or other more sophisticated measures~\cite{schreiber2000measuring, seth2005causal} of time series, Bayesian inference of graphical models~\cite{abbeel2006learning,salakhutdinov2008quantitative,montanari2009graphical,salakhutdinov2010deep,bresler2013reconstruction,amin2018quantum} and models of dynamics on networks~\cite{peixoto2019network}, among others.
These techniques are widely used (e.g., in neuroscience~\cite{hinne2013bayesian, breakspear2017dynamic, bassett2018nature}, genetics~\cite{wang2006inferring}, epidemiology~\cite{peixoto2019network, prasse2020network} and finance~\cite{musmeci2013bootstrapping}) to reconstruct interaction networks on which network science tools can be applied.

Interestingly, dynamics prediction and network reconstruction are usually considered separately, even though they are related to one another.
The emergent field of the network neuroscience~\cite{bassett2017network, sporns2020structure} is perhaps the most actively using both notions: Network reconstruction for building brain connectomics from functional time series, then dynamics prediction for inferring various brain disorders from these connectomes~\cite{fornito2015connectomics,van2019cross}.
Recent theoretical works have also taken advantage of these notions to show that dynamics hardly depend on the structure.
In Ref.~\cite{prasse2022predicting}, it was shown that time series generated by a deterministic dynamics evolving on a specific graph can be accurately predicted by a broad range of other graphs.
These findings highlight how poor our intuition can be with regard to the relationship between predictability and reconstructability.
Furthermore, recent breakthroughs in deep learning on graphs have benefited from proxy network substrates to enhance the predictive power of their models~\cite{zhang2018deep, zhou2018graph, xu2018powerful}, with applications in epidemiology~\cite{shah2020finding, murphy2021deep}, and pharmaceutics~\cite{fout2017protein, zitnik2018modeling}.
However, the use of graph neural networks and those proxy network substrates is only supported by numerical evidence and lacks a rigorous theoretical justification.
As a result, their enhanced predictability remains to be fully corroborated.
There is therefore a need for a solid, theoretical foundation of reconstructability, predictability and their relationship in networked systems.

In this work, we establish a rigorous framework that lays such a foundation based on information theory.
Information theory has been regularly applied to networks and dynamics in the past.
In network science, it has been used to characterize random graph ensembles~\cite{bianconi2009entropy, anand2009entropy, anand2010gibbs}---e.g. the configuration model~\cite{johnson2010entropic, anand2011shannon} and stochastic block models~\cite{peixoto2012entropy,  young2017finite}---, to develop network null models~\cite{cimini2019statistical} and to perform community detection~\cite{peixoto2014hierarchical, peixoto2017nonparametric}.
In stochastic dynamical systems, information-theoretical measures have been proposed to quantify their predictability~\cite{delsole2007predictability, kleeman2011information, scarpino2019predictability}, complexity~\cite{crutchfield1989inferring, feldman1998measures} and emergence \cite{rosas2020reconciling}. In statistical mechanics, information transmission has been shown to reach a maximum value near the critical point of spin systems in equilibrium~\cite{matsuda1996mutual, gu2007universal}.

Our objective is to combine these ideas into a single framework, motivated by recent works involving spin dynamics on lattices~\cite{barnett2013information, meijers2021behavior} and deterministic dynamics~\cite{prasse2022predicting}.
Our contributions are fourfold.
First, we use mutual information between structure and dynamics as a foundation for our general framework to quantify the structure-function relationship in complex systems. 
Second, this codependence naturally leads to the definition of measures of predictability and reconstructability.
Doing so allows us to conceptually \emph{unify} prediction and reconstruction problems, i.e., two classes of problems that are usually treated separately.
Third, we design efficient numerical techniques for evaluating these measures on large systems.
Finally, we identify a new phenomenon---a \emph{duality}---where our prediction and reconstruction capabilities can vary in opposite directions.
These findings further our understanding of the complexity of modeling networked complex systems, such as the brain, where both prediction and reconstruction techniques play critical roles.

\begin{figure}
    \centering
    \includegraphics[width=0.48\textwidth]{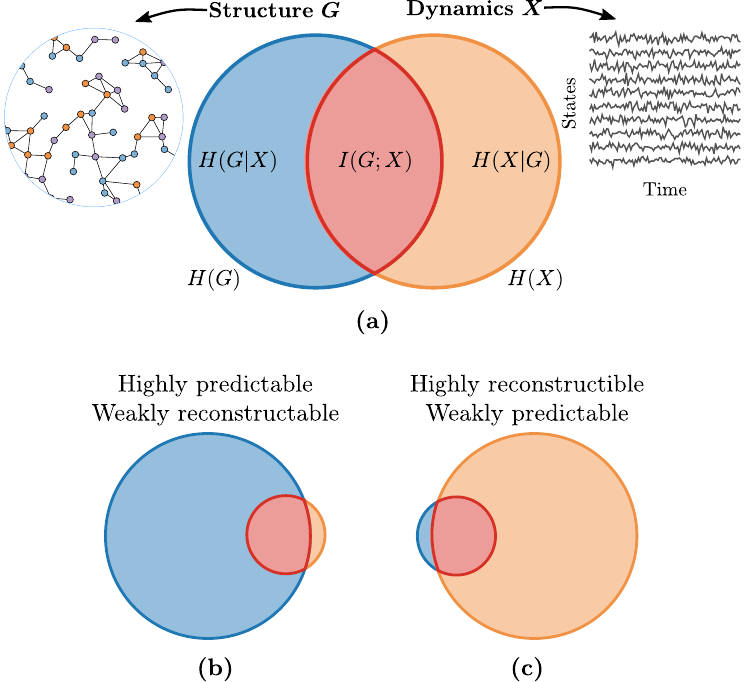}
    \caption{
    \textbf{Information diagram of dynamics on random graphs}.
    (a) Areas represent amounts of information: The entropies related to $G$ are shown on the left in blue and those related to $X$ are on the right in orange. Mutual information, in red, corresponds to the information shared by both $G$ and $X$.
    (b) The highly predictable / weakly reconstructable scenario, where $H(X) \gg H(G)$ meaning that $\mi$ contains most of the information related to the dynamics, but only a small fraction of the information related to the graph.
    (c) The reverse scenario, i.e., highly reconstructable / weakly predictable, where $H(G) \gg H(X)$ meaning that $\mi$ contains most of the information related to the graph, but only a small fraction of the information related to the dynamics..
    }
    \label{fig:info_diagram}
\end{figure}

\section{Results}

\subsection{Information theory of dynamics on random graphs}
\label{sec:information-theory}
Let us consider a random graph $G$ whose support, $\mathcal{G}_N$, consists in the set of all graphs of $N$ vertices, each of which having its respective non-zero prior probability $P(G^*)$ with $G^*\in\mathcal{G}_N$.
From the Bayesian perspective, the random graph $G$ represents our prior knowledge on the structure of the system of interest.
We also consider a stochastic process (also called a dynamics hereafter) of length $T$, noted $X$, evolving on a realization of $G$ and representing the possible dynamic states of the system.
We note $P(X\mid G)$ the probability of a random time series $X = (X_{i,t})_{i,t}$ conditioned on $G$, where $X_{i,t}$ is the random state, with support $\Omega$, of vertex $i$ at time $t$.
Together, $X$ and $G$ form a Bayesian chain $G\to X$, where the arrow indicates conditional dependence~\cite{edwards2012introduction}.

We are interested in the mutual information between $X$ and $G$---denoted $\mi$---which is a symmetric measure that quantifies the codependence between the dynamics $X$ and the structure $G$~\cite{cover2006elements},  where $\mi = 0$ when they are independent.
It is equivalently given by
\begin{subequations}\label{eq:mutual_info}
    \begin{align}
        \mi &= H(X) - H(X \mid G) \\ 
        &= H(G) - H(G \mid X)\,,
    \end{align}
\end{subequations}
where $H(G) = -\mean{\log P(G)}$ and $H(X) = -\mean{\log P(X)}$ are respectively the marginal entropies of $G$ and $X$, and $H(G\mid X) = -\mean{\log P(G\mid X)}$ and $H(X\mid G) = -\mean{\log P(X\mid G)}$ are their corresponding conditional entropies.
In the previous equations, the marginal distribution for $X$, the \emph{evidence}, is defined as $P(X) = \sum_{G^*\in\mathcal{G}_N} P(G^*)P(X\mid G^*)$, and the \emph{posterior} is obtained from Bayes' theorem as $P(G\mid X) = P(G)P(X\mid G)/P(X)$, using the given graph \emph{prior} $P(G)$ and the dynamics \emph{likelihood} $P(X\mid G)$.
In the case where $\Omega$ is a countable set (i.e., vertices have discrete dynamical states), $\mi$ is a non-negative measure bounded by $0\leq \mi \leq \min\set{H(G), H(X)}$.
Figure~\ref{fig:info_diagram}(a) provides an illustration of Eq.~\eqref{eq:mutual_info}.

\begin{figure}
    \centering
    \includegraphics[scale=1]{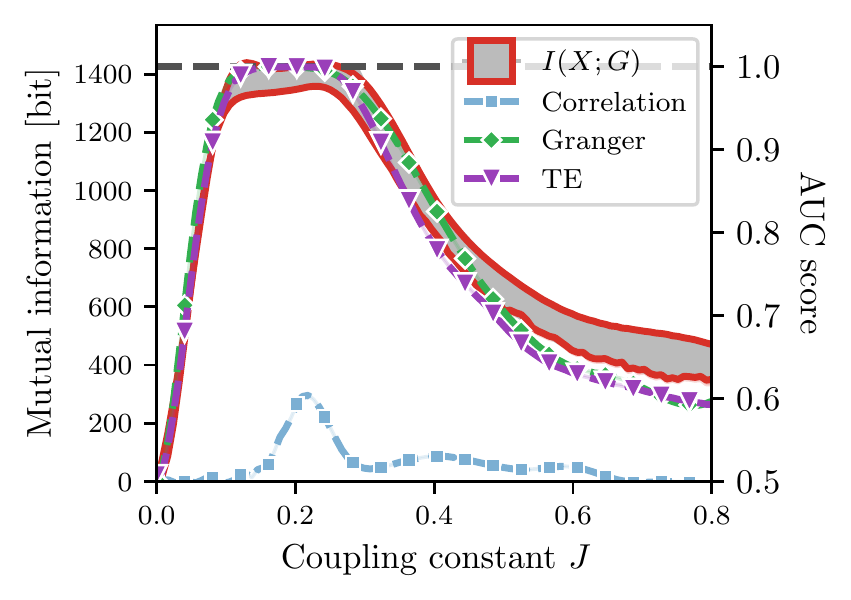}
    \caption{Comparison between the mean field (MF) estimator of $\mi$ (red circles) and the average area under the ROC curve (AUC) of the correlation matrix method~\cite{kramer2009network} (light blue squares), the Granger causality method \cite{schreiber2000measuring} (green diamonds) and the transfer entropy method \cite{seth2005causal} (purple triangles), as a function of the coupling constant for time series of length $T=100$ generated with the Glauber dynamics on Erdos-Renyi of size $N=100$ with $M=250$ edges. The mutual information is shown on the left axis and the AUC scores, on the right axis. The theoretical maximum values for both measures are shown with the one vertical dashed line. }
    \label{fig:mi-vs-heuristics}
\end{figure}

The measures presented in Eq.~\eqref{eq:mutual_info} and above can all be interpreted in the context of information theory. 
Information is generally measured in bits which in turn is interpreted as a minimal number of binary---i.e., yes/no---questions needed to convey it.
While entropy measures the uncertainty of random variables like $X$ and $G$, i.e., the minimal number of bits of information needed to determine their value, mutual information represents the reduction in uncertainty about one variable when the other is known.
The fact that it is symmetric means that this reduction goes both ways: The reduction in the dynamics uncertainty when the structure is known is equal to that of the structure when the dynamics is known.
Hence, mutual information measures the amount of information shared by both $X$ and $G$.

As an illustration, let us consider the physical example of a spin system such as the one described by the Glauber dynamics~\cite{glauber1963time} depending on $G$ (see Table~\ref{tab:dynamics}).
For a given value of the coupling parameter $J\geq0$, the spins will be more (large $J$) or less (small $J$) likely to align with their first neighbors in $G$. 
Now, suppose that $J=0$, which means that the spins will flip independently from each other and from $G$ with probability $\frac12$ at each time step. 
Hence, $H(X\mid G) = NT$ bits, corresponding to the maximum entropy of $X$: We need precisely one binary question for each spin at each time for a given structure $G$---e.g., ``\emph{Is the spin of vertex $i$ at time $t$ up?}".
When $J>0$, correlation is introduced between connected spins. 
As a result, a single question about the spin of vertex $i$ at time $t$ can provide additional information about the spins of other vertices at other times and thus, $H(X\mid G) < NT$.
The interpretation of $H(X)$ is analogous to that of $H(X\mid G)$, as it measures the number of binary questions needed to determine $X$ when the graph is unknown.
From this perspective, the mutual information $\mi$, as expressed by the difference between $H(X)$ and $H(X\mid G)$, is the reduction in the number of questions needed to predict $X$ ensuing from the knowledge of $G$. 
Hence, $\mi$ measures how much information about $X$ is determined by $G$---or how influential $G$ is over $X$---and is therefore related to its \emph{predictability}.

Similar observations can be made from the structural perspective. 
Suppose that $X$ is again the Glauber dynamics and $G$ is a random graph, where each edge exists independently with probability $p$. 
This yields $H(G) = -\binom{N}{2} [p\log p + (1 - p) \log (1 - p)]$, where $\binom{N}{2}$ is the total number of possible undirected edges.
When $p=\frac12$, we have $H(G) = \binom{N}{2}$ bits, which is again the maximum entropy of $G$. 
We therefore need precisely one binary question for each of the $\binom{N}{2}$ edges in the graph---e.g. ``\emph{Is there an edge between $i$ and $j$?}"---to completely determine its value.
When the dynamics $X$ is known, $H(G\mid X)$ is interpreted similarly to $H(G)$, but also takes into account the observation of the spins $X$ which introduces correlation between the edges of $G$. 
As a result, each bit can provide information about more than one edge, even in the case $p=\frac12$ where we a priori need one bit per possible edge to fully reconstruct $G$.
Consequently, the knowledge of $X$ reduces uncertainty about $G$ (i.e., $H(G\mid X) \leq H(G)$, see \cite[Theorem 2.6.5]{cover2006elements}). 
The difference between $H(G)$ and $H(G\mid X)$---i.e., $\mi$---thus measures how much information about $G$ is revealed by knowing $X$, which in turn is related to its \emph{reconstructability}.

In practice, we can argue that $\mi$ is related to the performance of reconstruction algorithms such as the cross-correlation matrix method~\cite{kramer2009network}.
Figure~\ref{fig:mi-vs-heuristics} provides evidence of this relationship by comparing the performance of common reconstruction algorithms with the mutual information (see Section~\ref{app:estimators} for detail).
Indeed, when $I(G;X) = 0$, the score of all algorithms is comparable to that of a random edge/no-edge classifier between each pair of nodes (with an AUC of 0.5) and all methods seems to peak around the mutual information maximum before decreasing again for larger coupling.
Note that a similar comparison analysis could, in principle, be carried out for the predictability as well, but the problem of measuring a graph influences over a process is less documented than that of graph reconstruction, which makes it harder to actually investigate.

The mutual information $\mi$ is therefore both a measure of predictability and reconstructability, thereby unifying these two concepts under one single framework.
We say that a system is perfectly predictable when the mutual information contains all the information about $X$, that is when $\mi = H(X)$ [see Fig.~\ref{fig:info_diagram}(b)].
Likewise, we say that it is perfectly reconstructable when $\mi = H(G)$ [see Fig.~\ref{fig:info_diagram}(c)].
Consequently, whenever $\mi>0$, we expect the system to be predictable and reconstructable to a certain degree.
Otherwise, when $\mi = 0$, the system is said both unpredictable and unreconstructable.
Yet, $\mi$ by itself is hardly comparable from one system to another. 
Indeed, a specific value of $I(G;X)$ may correspond to opposing scenarios when it comes to predictability and reconstructability, as shown in Fig.~\ref{fig:info_diagram}(b-c). 
Thus, it is more convenient to use normalized quantities such as the uncertainty coefficients
\begin{subequations}
    \begin{align}
        \pred &= \frac{\mi}{H(X)}\,,\label{eq:pred}\\
        \recon &= \frac{\mi}{H(G)} \,, \label{eq:recon}
    \end{align}
\end{subequations}
as measures, bounded between 0 and 1, of the relative degrees of \emph{predictability} and \emph{reconstructability}, respectively.

The interpretation of the reconstructability measure $U(G\mid X)$ is straightforward: It is the fraction of the structural information that can be recovered from the dynamics.
Similarly, the predictability measure $U(X\mid G)$ is interpreted as the fraction of dynamical information that is determined---or influenced---by the structure.
Strictly speaking, $U(X\mid G)$ is not a standard measure of predictability, i.e., a measure that quantifies the influence of the initial conditions of a system, or its past states, over its future states \cite{Feder1994,delsole2007predictability, kleeman2011information, giannakis2012information}.
However, as the graph can be interpreted as an element of the initial conditions that remains constant throughout the process, our predictability measure $U(X\mid G)$ is compatible with the terminology of ``predictability".
It is nevertheless possible to utilize our framework towards predictability measures explicitly dependent on the past, but these do not highlight the relationship between structure and dynamics as clearly as the measures presented in this section.
We refer to Appendix~\ref{app:past_mi} for further details.

\begin{figure*}
    \centering
    \includegraphics[scale=1]{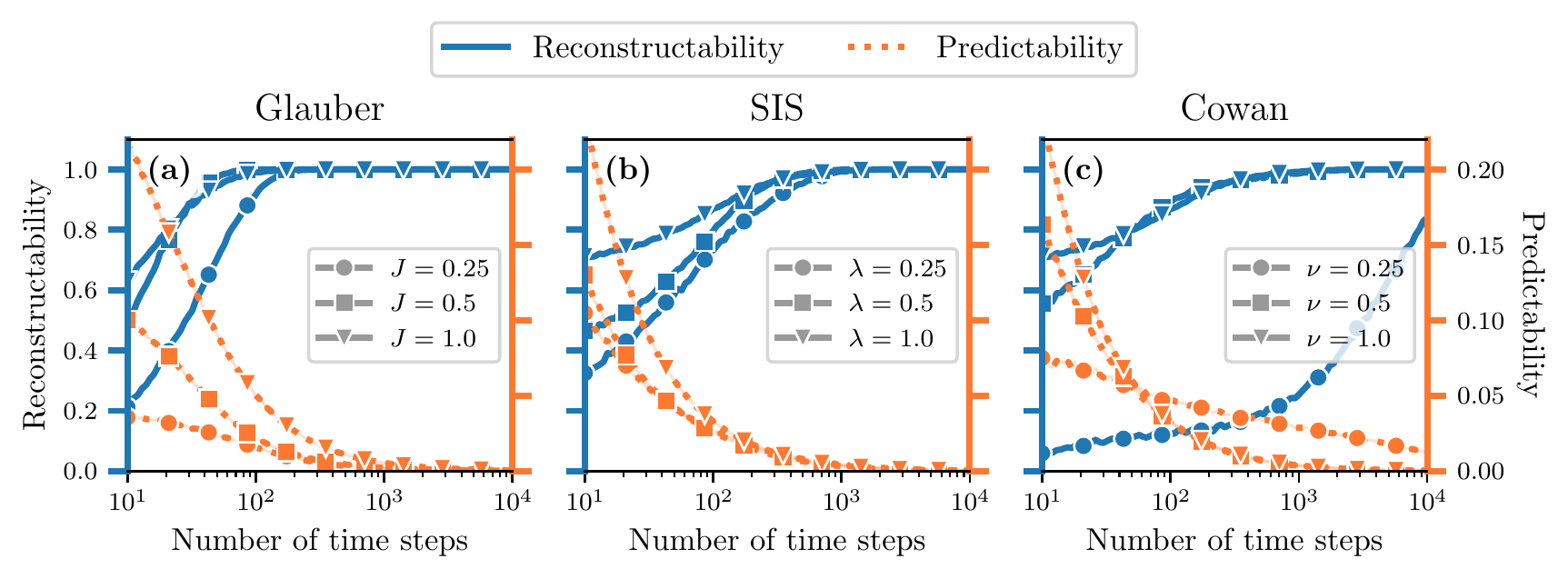}
    \caption{\textbf{$T$-duality in binary dynamics evolving on small Erd\H{o}s-R\'enyi random graphs}:
    (a) Glauber dynamics,
    (b) SIS dynamics and
    (c) Cowan dynamics.
    Each panel shows the reconstructability coefficient $\recon\in[0,1]$ (blue) and the predictability coefficient $\pred\in[0,1]$ (orange) as a function of the number of time steps $T$.
    We used graphs of $N=5$ vertices and $E=5$ edges, meaning an average degree of $\avgk = 2$.
    Each symbol corresponds to the average value measured over $1000$ samples.
    We also show different values of the coupling parameters---normalized by the average degree---using different symbols: (a) $J\avgk\in\set{\frac{1}{2}, 1, 2}$ for Glauber, (b) $\lambda\avgk\in\set{1, 2, 4}$ for SIS and (c) $\nu\avgk\in\set{1, 2, 4}$ for Cowan.
    }
    \label{fig:exact-duality-timestep}
\end{figure*}

\subsection{\texorpdfstring{$\theta$}{Theta}-Duality between predictability and reconstructability}
\label{sec:duality}
Predictability and reconstructability in dynamics on random graphs offer two perspectives of the same information shared by $G$ and $X$---two sides of the same coin.
However, it does not mean that predictability and reconstructability go hand in hand even though they are related: A high value of $\recon$ does not necessarily imply a high value of $\pred$, which can somewhat be counterintuitive.
In other words, a maximally influential structure, with respect to the dynamics, will not necessarily be easier to reconstruct.
This observation is well illustrated by Figs.~\ref{fig:info_diagram}(b)--(c), where $\recon$ and $\pred$ can take opposing values, depending on $H(G)$ and $H(X)$, for a same value of $\mi$.

As an example, let us consider $X$ to be a Markov chain evolving on a random graph $G$ for different values of the number of time steps, $T$.
Theorem \ref{thm:universality} (see App.~\ref{app:universality}) states that, for any Markov chain whose entropy rate is non-zero and for sufficiently large $T$, $\recon$ is an increasing function of $T$, while $\pred$ is a decreasing one.
This is a consequence of the fact that the mutual information is strictly increasing with $T$, and so is $\recon$ whenever $H(G)$ is independent of $T$. 
Yet, we show in App.~\ref{app:universality} that $\mi$ increases more slowly than $H(X)$ with $T$, which results in a decreasing $\pred$.
We refer to this opposing behavior as a \emph{duality} between $\recon$ and $\pred$ with respect to $T$, or a $T$-duality for short \footnote{Not to be confused with target space duality in string theory~\cite{giveon1994target}.}.

Figure~\ref{fig:exact-duality-timestep} illustrates the universality of the $T$-duality using different binary Markov chains (i.e., $\Omega = \set{0, 1}$).  In each of these chains, the probability $P(X\mid G)$ is
\begin{equation}
    P(X\mid G) = P(X_1)\prod_{t=1}^{T-1} P\left(X_{t + 1} \mid X_{t}, G\right)\,,
\end{equation}
where
\begin{widetext}
    \begin{equation}
        \begin{split}
            P\left(X_{t + 1} \mid X_{t}, G\right) = \prod_{i=1}^N \bigg\{
            &\big[\alpha(n_{i,t}, m_{i,t})\big]^{(1 - X_{i, t}) X_{i, t + 1}}
            \big[1 - \alpha(n_{i,t}, m_{i,t})\big]^{(1 - X_{i, t}) (1 - X_{i, t + 1})} \\
            &\big[\beta(n_{i,t}, m_{i,t})\big]^{X_{i, t} (1 - X_{i, t + 1})}
            \big[1 -\beta(n_{i,t}, m_{i,t})\big]^{X_{i, t} X_{i, t + 1}}
            \bigg\}
        \end{split}
    \end{equation}
\end{widetext}
is the transition probability from state $X_t$ to state $X_{t + 1}$.
We also denote the activation ($0\to1$) and the deactivation ($1\to0)$ probability functions with $\alpha(n_{i, t}, m_{i, t})$ and $\beta(n_{i, t}, m_{i, t})$, respectively, where $n_{i, t}$ and $m_{i, t}$ denote the number of active and inactive neighbors of vertex $i$ at time $t$.

We consider three well-known Markov chain models of different origins: The Glauber dynamics, the Suspcetible-Infectious-Susceptible (SIS) dynamics and the Cowan dynamics.
The aforementioned Glauber dynamics~\cite{glauber1963time}, which have been used to describe the time-reversible evolution of magnetic spins aligning in a crystal, have been tremendously studied because of its critical behavior and its phase transition.
Its stationary distribution is given by the Ising model which has found many applications in condensed-matter physics~\cite{mezard2009information} and statistical machine learning~\cite{binder2010monte,edwards2012introduction}, to name a few.
The SIS dynamics is a canonical model in network epidemiology \cite{pastor2015epidemic} often used for modeling influenza-like disease \cite{anderson1992infectious}, where periods of immunity after recovery are short.
In this model, susceptible (or inactive) vertices get infected by each of their infected (active) first neighbors, with a constant transmission probability, and recover from the disease with a constant recovery probability.
The simplicity of the SIS model has allowed for deep mathematical analysis of its absorbing-state phase transition \cite{pastor2001epidemic, ferreira2012epidemic, st2018phase}.
Finally, the Cowan dynamics~\cite{cowan1990stochastic} has been proposed to model the neuronal activity in the brain.
In this model, quiescent neurons fire if their input current, coming from their firing neighbors, is above a given threshold.
Its mean-field approximation \cite{Painchaud2022} reduces to the Wilson-Cowan dynamics \cite{Wilson1972}, one of the most influential models in neuroscience \cite{destexhe2009wilson}.
For each model, we can identify an inactive state---down, susceptible or quiescent---and an active one---up, infectious or firing.
The corresponding activation and deactivation probabilities are given in Table~\ref{tab:dynamics}.

Figure~\ref{fig:exact-duality-timestep} numerically supports Theorem~\ref{thm:universality} and clearly illustrates the $T$-duality for each dynamics and with different values of their parameters.
We used the Erd\H{o}s-R\'enyi model as the random graph on which these dynamics evolve. The support $\mathcal{G}_N$ is the set of all simple graphs of $N$ vertices with $E$ edges, and
\begin{equation}
    P(G) = \binom{\binom{N}{2}}{E}^{-1}\,.
\end{equation}

It is also important to note that the $T$-duality seems to persist for the past-dependent measures presented in Appendix~\ref{app:past_mi}, as illustrated by Fig.~\ref{fig:figure6}, for different values of $\tau$, which we recall is the length of the past Markov chain. 
However, it does not hold for all values of $\tau$, especially those that scale with $T$ such that $\tau = T - \xi$, where $\xi < T$ is constant.
Hence, it is tantalizing to conjecture that there exists a scaling $g(T)$ such that, when $\tau$ is dominated by $g(T)$, the $T$-duality can persist and it cannot otherwise.
More details are available in Appendix~\ref{app:past_mi}.

\begin{table}[t]
\setlength\tabcolsep{6pt}
    \begin{tabular}{c c c c}
        \hline \hline
        Dynamics & $\alpha(n, m)$ & $\beta(n, m)$ & Coupling\\ \hline \\
        Glauber~\cite{glauber1963time} & $\sigma[2J(n - m)]$ & $\sigma[2J(m - n)]$ & $J$\\ \\
        SIS~\cite{pastor2015epidemic} & $1 - \left(1 - \frac{\lambda}{\beta}\right)^m$ & $\beta$ & $\lambda$\\ \\
        Cowan~\cite{cowan1990stochastic}& $\sigma[a( \nu m - \mu)]$ & $\beta$ & $\nu$ \\\\ \hline\hline
    \end{tabular}
    \caption{Activation and deactivation probability functions, $\alpha(n, m)$ and  $\beta(n, m)$, respectively, for the binary dynamics considered in this study, where $n$ corresponds to the number of inactive neighbors whose states are $0$, and $m$ corresponds to the number of active neighbors whose states are $1$. We define $\sigma(x) = [\exp(-x) + 1]^{-1}$ as the logistic function. Some of these parameters are fixed throughout the paper: $\beta = 0.5$ for SIS and Cowan, and $a = 7$ and $\mu = 1$ for Cowan. The coupling parameters ($J$ for Glauber, $\lambda$ for SIS and $\nu$ for Cowan) are specified in each figure. Also, to prevent the SIS dynamics from being completely inactive, we allow the inactive vertices to spontaneously activate with probability $\epsilon = 10^{-3}$~\cite{van2012epidemics}.}
    \label{tab:dynamics}
\end{table}

The observation of the $T$-duality begs for a more general definition of duality for any arbitrary parameter $\theta$ (see Appendix~\ref{app:duality}).
In fact, we say that $\recon$ and $\pred$ are \emph{dual} with respect to $\theta$, or \emph{$\theta$-dual}, in an interval $\Theta$ if and only if the signs of their derivative with respect to $\theta$ are different for every $\theta^{*}\in\Theta$:
\begin{equation}
    \left[\pder{\recon}{\theta}\pder{\pred}{\theta}\right]_{\theta=\theta^*} < 0\,.
\end{equation}
This criterion formally relies on the existence of regions $\Theta$ where the variations of $\recon$ and $\pred$ with respect to $\theta$ are contradictory, regardless of their amplitude.
We use this criterion to relate the existence of extrema of $\recon$ and $\pred$ with that of regions of $\theta$-duality (see Lemma~\ref{lem:extrema} in App.~\ref{app:duality}), and to prove Theorem~\ref{thm:universality}.

Knowing the existence of the $T$-duality and having a general definition of $\theta$-duality, it is now natural to ask if there exist other types of $\theta$-dualities in dynamics on random graphs. 
A large variety of parameters could lead to interesting $\theta$-dualities---some controlling the general behavior of the dynamics, and others controlling some structural properties of the random graph which, in turn, also impact the dynamics. 
Most of them require the system to be larger, if the effects over $X$ and $G$ of varying $\theta$ are to be significant (e.g. phase transitions). 
However, in high-dimensional systems, theoretical and numerical challenges arise in the evaluation of the reconstructability and the predictability, which complicate the search for dualities. We address this problem in the next section.

\begin{figure*}
    \centering
    \includegraphics[scale=1]{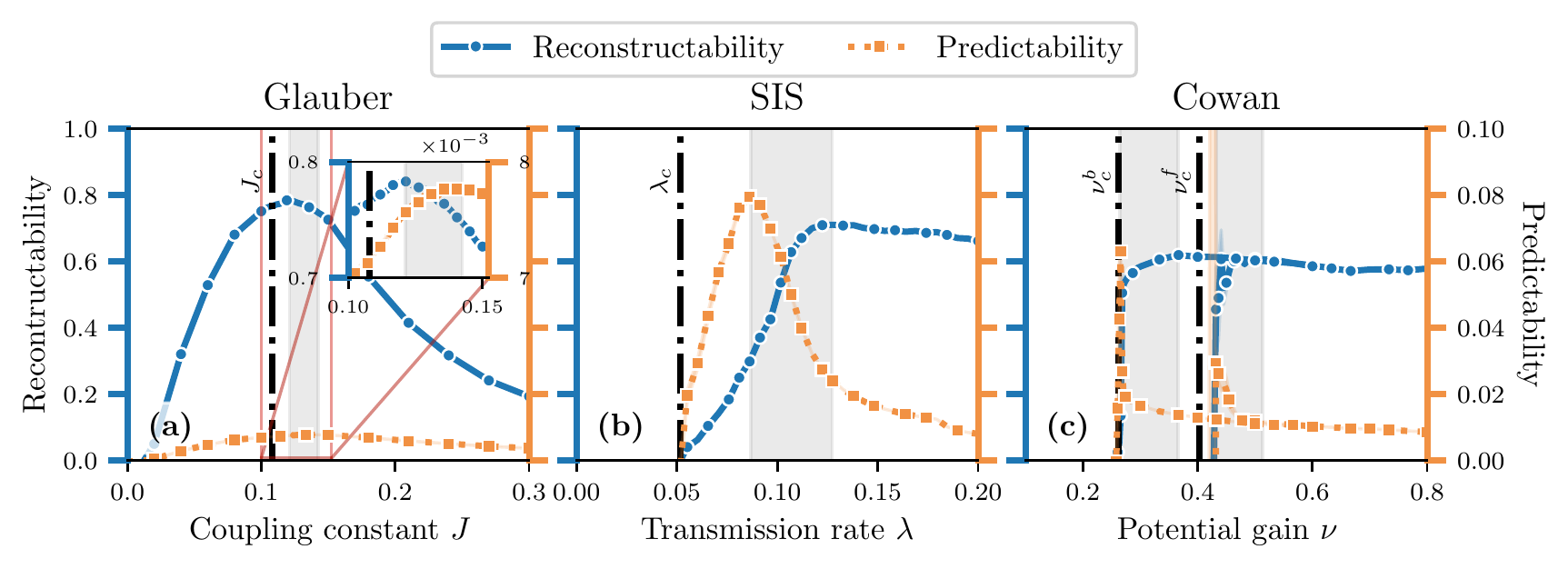}
    \caption{\textbf{Dynamics evolving on configuration model graphs with geometric degree distribution}:
    (a) Glauber dynamics,
    (b) SIS dynamics and
    (c) Cowan dynamics.
    We generated graphs of $N=1000$ vertices, where the degree distribution $\rho(k) = (1 - p)p^k$ is a geometric distribution with $p = \mean{k}/(1 + \mean{k})$ and $\mean{k}=5$, and time series of length $T=2000$. See Table~\ref{tab:dynamics} for the remaining parameters.
    Similar to Fig.~\ref{fig:exact-duality-timestep}, $\recon$ is shown in blue (left axis) and $\pred$ is shown in orange (right axis).
    We show, for each dynamics, the uncertainty coefficients as a function of the coupling parameter: $J$ for Glauber, $\lambda$ for SIS and $\nu$ for Cowan.
    The vertical dotted-dashed lines correspond to the phase transition thresholds of each dynamics, which are estimated from Monte Carlo simulations (see Appendix~\ref{app:thresholds}).
    For the Cowan dynamics, the forward and backward branches are shown with their corresponding thresholds and dual regions (see main text).
    }
    \label{fig:approx-duality-coupling}
\end{figure*}

\subsection{Duality and criticality}
Despite their different nature and range of applications, the three models presented in Table~\ref{tab:dynamics} share several properties of interest.  For instance, each model has a coupling parameter that controls the influence of the state of the first neighbors on the transition probabilities.  They also all feature a phase transition in the infinite size limit whose position is determined by the coupling parameter (see Fig.~\ref{fig:figure5} and App.~\ref{app:thresholds}).
We now investigate the influence of criticality over the existence of $\theta$-dualities, where $\theta$ is a coupling parameter.

For the Glauber dynamics, this parameter is the coupling constant $J$, which dictates the reduction (increase) in the total energy of a spin configuration when two neighboring spins are parallel (antiparallel).
The Glauber dynamics features a continuous phase transition at a critical point $J_{\mathrm{c}}$ between a disordered and an ordered phase, where for $J<J_{\mathrm{c}}$ the spins are disordered resulting in a vanishing magnetization, and for which this magnetization is non-zero when $J>J_{\mathrm{c}}$. 
For the SIS dynamics, it is the transmission rate $\lambda$ that acts as a coupling parameter. 
Like the Glauber dynamics, the SIS dynamics possesses a continuous phase transition where, when $\lambda<\lambda_{\mathrm{c}}$, the system reaches an absorbing---or inactive---state from which it cannot escape, and an active state, when $\lambda > \lambda_{\mathrm{c}}$, where a non-zero fraction of the vertices remain active over time
\footnote{
    It is not strictly accurate to say that our considered version of the SIS dynamics reaches a true absorbing state, since we allow for self-infection $\epsilon$ which allows it to escape the completely inactive state. 
    Instead, it reaches a metastable state where most of the vertices are asymptotically inactive. 
    However, it can be shown that the two phase transitions are quite similar for small $\epsilon$~\cite{van2012epidemics}.
}.
The Cowan dynamics can both feature a continuous or a first-order phase transition between an inactive and an active phase depending on the value of slope $a$, for which the coupling parameter is $\nu$, i.e., the potential gain for each firing neighbors.
The continuous and first-order phase transitions of the Cowan dynamics are quite different in that the latter is characterized by two thresholds, namely the forward and backward thresholds $\nu^{\mathrm{b}}_{\mathrm{c}} < \nu^{\mathrm{f}}_{\mathrm{c}}$, respectively (see Appendix~\ref{app:thresholds} for further details).
Hence, the Cowan dynamics has a first-order phase transition that exhibits a bistable region $\nu\in( \nu^{\mathrm{b}}_{\mathrm{c}},  \nu^{\mathrm{f}}_{\mathrm{c}})$, where both the inactive and active phases are reachable depending on the initial conditions.

To account for the heterogeneous network structure observed in a wide range of complex systems~\cite{barabasi2013network}, we simulate the dynamics on the configuration model, a random graph whose---potentially heterogeneous---degree sequence $\vec{k}$ is fixed and whose support $\mathcal{G}_N$ corresponds to the set of all loopy multigraphs of degree sequence $\vec{k}$. The probability of a graph $G^*$ in this ensemble is
\begin{equation}
    P(G^*) = \frac{(2E)!!}{(2E)!} \frac{\prod_{i} k_i!}{\prod_{i<j}{M_{ij}!} \prod_i M_{ii}!!}\,,
\end{equation}
where $M_{ij}$ counts the number of edges connecting vertices $i$ and $j$ in the multigraph $G^*$ and $2E = \sum_{i}k_i$ is the number of half-edges in $G^*$.
Like the Erd\H{o}s-R\'enyi model, the configuration model fixes the number of edges, but also fixes the degree distribution $\rho(k)$.

Figure~\ref{fig:approx-duality-coupling} shows the predictability and reconstructability of the three dynamics evolving on graphs drawn from the configuration model whose distribution, $\rho(k)$, is geometric, as estimated by the MF estimator.
First, these results allow us to compare the dynamics with one another.
For example, on the one hand, the Glauber dynamics is globally less predictable than the other two, since its predictability coefficient is overall smaller.
In other words, the knowledge of a graph $G^*$ provides less information about $X^*$ in the Glauber dynamics in comparison with the others, relatively to the total amount of information needed to reconstruct $X^*$.
This is related to the time reversibility of the Glauber dynamics, which allows any vertex to transition from the inactive to the active state (and vice versa) with non-zero probability, at any time, effectively making the Glauber dynamics more random than the others---i.e. $H(X)$ is greater for Glauber than the other processes.
On the other hand, the SIS and Cowan dynamics are portrayed by the MF estimator as practically unpredictable and unreconstructable when their coupling parameter is below their respective critical point. 
This precisely occurs in the inactive phase, where no mutual information can be generated after a short time, when the system reaches the inactive state.
By contrast, the Glauber dynamics does not reach an inactive state below its critical point, which explains the gradual increase in predictability and reconstructability in that region.

Several additional observations are worth making.
All dynamics exhibit maxima for $\pred$ and $\recon$ which delineate a region of duality illustrated by the shaded areas (two for Cowan, that is one for each branch).
These regions are close to, but systematically above, their respective phase transition thresholds.
A similar phenomenon in spin dynamics on non-random lattices has been reported by previous works~\cite{barnett2013information,meijers2021behavior}, in which the information transmission rate between spins---a measure akin to $\mi$---is maximized above the critical point.
Our numerical results are consistent with theirs, and suggest that their findings regarding near-critical systems even apply beyond spin dynamics on fixed lattices, to other types of processes on more heterogeneous and random structures.

\section{Discussion}
In this work, we used information theory to characterize the structure-function relationship with mutual information.
We showed how mutual information is a natural starting point to define both predictability and reconstructability in dynamics on networks, in turn showing how they are intrinsically unified.
Our approach is quite general allowing the exploration of different configurations of dynamics on networks of the form $G\to X$, thus varying the nature of the process itself as well as the random graph on which it evolves.
Our framework could be extended to adaptive systems~\cite{gross2008adaptive, marceau2010adaptive, scarpino2016effect, Khaledi-Nasab2021} where both $X$ and $G$ influence each other (i.e., $X\leftrightarrow G$).
The relationship between $X$ and $G$ could also go the other way around: A system in which $X$ generates a graph $G$ (i.e., $X\to G$).
Hyperbolic graphs~\cite{krioukov2010hyperbolic, boguna2021network} falls into this category, where $X$ represents a set of coordinates, and our framework could be extended to quantifying the feasibility of network geometry inference~\cite{boguna2010sustaining, papadopoulos2015network, garcia2019mercator}.

We found efficient ways to estimate the mutual information numerically, thus allowing us to investigate relatively large systems.
More work on this front is required, however, since the evaluation of these estimators remains quite computationally costly.
It would be worth investigating simpler models, for which it is possible to analytically---or at least approximately---evaluate $\pred$ and $\recon$.
In particular, dimension reduction methods~\cite{laurence2019spectral, thibeault2020threefold, thibeault2022boran} and approximate master equations~\cite{gleeson2011high, st-onge2021pre} are promising avenues for obtaining reliable approximations of $\mi$, $\pred$ and $\recon$.

Central to our findings is the peculiar discovery that predictability and reconstructability are not only related, but sometimes dual to one another.
We proved that such $\theta$-duality appears when the length of the processes changes and presented numerical evidence of duality near the criticality in three different dynamics on random heterogeneous networks.
These findings generalize and formalize---while being consistent with---previous works~\cite{barnett2013information, meijers2021behavior} and suggest that criticality in these systems is intrinsically related to the duality.

From a practical perspective, the existence of such a $\theta$-duality can be critical to network modeling applications, since it also suggests a predictability-reconstructability trade-off.
On the one hand, we can choose this parameter $\theta$ such that the uncertainty of the reconstructed structure is minimized, at the expense of having a less informative structure with respect to the dynamics.
On the other hand, we can consider the reverse case, where the process is maximally influenced by the inferred structure, whose uncertainty is nevertheless not minimized.
Analogous to the position-momentum duality in the Heisenberg uncertainty principle of quantum mechanics, the predictability-reconstructability duality must be accounted for in our network models if we are to disentangle complex systems.

\section*{Acknowledgments}
We are grateful to Guillaume St-Onge and Vincent Painchaud for useful comments, and to Simon Lizotte and François Thibault for their help in designing the software. This work was supported by the Fonds de recherche du Qu\'ebec -- Nature et technologies (VT, PD), the Conseil de recherches en sciences naturelles et en g\'enie du Canada (CM, VT, AA, PD), and the Sentinelle Nord program of Universit\'e Laval, funded by the Fonds d’excellence en recherche Apog\'ee Canada (CM, VT, AA, PD).
We acknowledge Calcul Qu\'ebec and Compute Canada for their technical support and computing infrastructures.


\newpage
\section{Materials and Methods}
\subsection{Formal definition of \texorpdfstring{$\theta$}{Theta}-duality}
\label{app:duality}
In what follows, we define the duality between predictability and reconstructability by taking a more general stance: Instead of considering a stochastic process $X$ evolving on a random graph $G$, we let $X$ be conditioned on an arbitrary discrete random variable $Y$. First, we define the local duality of the uncertainty coefficients. The latter are considered as continuously differentiable functions with respect to a parameter $\theta$ whose domain is some non-empty interval of the real line.
%
\begin{definition}[Local duality]
    The uncertainty coefficients $U(X\mid Y)$ and $U(Y\mid X)$ are \emph{locally dual} with respect to $\theta$ at $\theta = \theta^*$ if and only if
    \begin{equation}\label{eq:duality-criterion}
        \left[\pder{U(X\mid Y)}{\theta} \pder{U(Y\mid X)}{\theta}\right]_{\theta=\theta^*} < 0\,.
    \end{equation}
\end{definition}
The definition of the $\theta$-duality, a global property, follows that of the local duality.
\begin{definition}[$\theta$-Duality]
    The uncertainty coefficients $U(X\mid Y)$ and $U(Y\mid X)$ are \emph{dual} with respect to $\theta$, or $\theta$\textit{-dual}, in the interval $\Theta$ if and only if they are locally dual for all values of $\theta^*$ in $\Theta$.
\end{definition}
From these definitions, we relate the presence of extrema of $U(X\mid Y)$ and $U(Y\mid X)$ with the existence of a $\theta$-duality.
\begin{lemma}\label{lem:extrema}
    Let $\Theta$ be a non-empty subinterval of the variable $\theta$ whose one endpoint is a local extremum of $U(X\mid Y)$ and the other, a local extremum of $U(Y\mid X)$.
    Moreover, suppose that $U(X\mid Y)$ and $U(Y\mid X)$ do not have critical points in $\Theta$. Then the extrema points delineate a region of $\theta$-duality if and only if they are both maxima (or both minima).
\end{lemma}
\begin{proof}
    Let $\theta_R$ and $\theta_P$ be the extrema points of $U(Y\mid X)$ and $U(X\mid Y)$, respectively. Thus
    \begin{equation}
        \left.\pder{U(Y\mid X)}{\theta}\right|_{\theta=\theta_R} = \left.\pder{U(X\mid Y)}{\theta}\right|_{\theta=\theta_P} = 0\,.
    \end{equation}
    Suppose for a moment that $\theta_R < \theta_P$ and let $\Theta=(\theta_R, \theta_P)$.
    This implies that $\pder{U(Y\mid X)}{\theta}$ changes sign at $\theta_R$, before $\pder{U(X\mid Y)}{\theta}$, for which the sign change happens at $\theta_P$.

    On the one hand, if the extrema points $\theta_R$ and $\theta_P$ are both maxima (or minima), then $\pder{U(Y\mid X)}{\theta}$ and $\pder{U(X\mid Y)}{\theta}$ have different signs in $\Theta$.
    Hence, inequality~\eqref{eq:duality-criterion} is verified in this region.
    The uncertainty coefficients are therefore $\theta$-dual in $\Theta$.
    
    On the other hand, if the uncertainty coefficients are $\theta$-dual in $\Theta$, then inequality~\eqref{eq:duality-criterion} is satisfied in this interval. This in turn implies that either $U(Y\mid X)$ decreases in $\Theta$ while $U(X\mid Y)$ increases or $U(Y\mid X)$ increases in $\Theta$ while $U(X\mid Y)$ decreases. Therefore, the endpoints of $\Theta$ are either both maximum points or both minimum points.  
    
    Finally, repeating the same arguments with $\theta_R > \theta_P$  and $\Theta=(\theta_P, \theta_R)$ leads to the same conclusions about $\theta$-duality of $U(X\mid Y)$ and $U(Y\mid X)$ in $\Theta$.
\end{proof}

\subsection{Universality of the \texorpdfstring{$T$}{T}-duality}
\label{app:universality}
We demonstrate the universality of the $T$-duality, where $T$ is the number of steps in the process $X$. First, we need to show that the mutual information is a monotonically increasing function of $T$.

\begin{lemma}\label{lem:monotonicity}
    Let $X=(X_1, X_2,\cdots, X_T)$ be a Markov chain of length $T$ whose transition probabilities are conditional to some discrete random variable $Y$ that is independent of $T$ and such that $H(X_{t+1} |X_{t})>0$ for all $t\in\{1,\ldots,T-1\}$. Suppose moreover that the state spaces of $X$ and $Y$ are finite. Then the mutual information $I(X;Y)$ is nonzero and monotonically increasing with $T\in\mathbb{Z}_+$.
\end{lemma}

\begin{proof}
    Let us define a Markov chain $X' = (X_1, X_2, \cdots, X_{T - 1})$ of size $T-1$, such that the concatenation of $X'$ with state variable $X_{T}$ yields $X$. Hence, we can express the mutual information between $X$ and $Y$ in terms of $X'$ as $I(X;Y) = I(X', X_T; Y)$. Furthermore, proving the monotonicity of mutual information can be reformulated as proving the following inequality:
    \begin{equation}\label{eq:mutualinfo_diff1}
        I(X', X_T; Y) - I(X';Y) > 0\,,
    \end{equation}
    for all $T$. By the chain rule for conditional mutual information, that is $I(X',X_T;Y) = I(X_T;Y|X') + I(X';Y)$, inequality \eqref{eq:mutualinfo_diff1} becomes
    \begin{equation}\label{eq:mutualinfo_diff2}
        I(X_T;Y|X') = H(X_T|X') - H(X_T\mid X', Y) > 0\,.
    \end{equation}
    The term $H(X_T|X') - H(X_T\mid X', Y)$ is always at least non-negative, by virtue of the non-negativity of mutual information~\cite[Theorem 2.6.5]{cover2006elements}. Then, to prove inequality \eqref{eq:mutualinfo_diff2}, we must verify that $H(X_T|X')$ never equals $H(X_T\mid X', Y)$. Recalling that $H(X_T\mid X') \geq H(X_T\mid X', Y) \geq 0$, inequality \eqref{eq:mutualinfo_diff2} does not hold if (\emph{i}) $H(X_T|X') = 0$ or if (\emph{ii}) $X_T$ is independent of $Y$ (i.e., $I(X_T;Y|X') = 0$).
    According to the hypothesis $H(X_{t+1} |X_{t})>0$ for all $t\in\{1,\ldots,T-1\}$, condition (\emph{i}) cannot be true.
    Moreover, condition (\emph{ii}) implies that $I(X;Y) = I(X_T, X';Y) = I(X';Y) = 0$. Therefore, the only instance where Eq.~\eqref{eq:mutualinfo_diff1} is not satisfied is when the Markov chain $X$ is independent of $Y$, i.e., $I(X;Y) = 0$ for all length $T$.
    However, this contradicts the assumption about the transition probabilities. 
    Hence, $I(X;Y)>0$ and monotonically increases with $T$.
\end{proof}

Before presenting the main result of this section, let us make a few remarks about the restrictions imposed in the last lemma. 
The condition $H(X_{t+1} |X_{t})>0$ for all $t\in\{1,\ldots,T-1\}$ only asserts that the Markov chain is nondeterministic in the sense that knowing the state of the chain at time $t$ does not completely eliminate the uncertainty about the state at time $t+1$.
This condition is satisfied for wide variety of stochastic processes, including the irreducible Markov chains, where there is always a nonzero probability to transition from a state to any other state in a finite number of time steps. 
Moreover, the finiteness of the state spaces for the chain $X$ and the variable $Y$ is imposed to make $H(X)$, $H(Y)$, and $I(X;Y)$ finite. This in turn ensures that the uncertainty coefficients $U(Y\mid X)$ and $U(X\mid Y)$ are well defined for all $T\in\mathbb{Z}_+$, a property that is necessary to prove the next lemma.  

\begin{lemma}\label{lem:extension}
    Let $X=(X_1, X_2,\cdots, X_T)$ and $Y$ respectively be a Markov chain and a discrete random variable as in Lemma \ref{lem:monotonicity}. Then the uncertainty coefficients $U(Y\mid X)$ and $U(X\mid Y)$, interpreted as functions of $T\in\mathbb{Z}_+$, can be uniquely generalized to functions, respectively $f(T)$ and $g(T)$, that are holomorphic for all $T\in\mathbb{C}$, and thus real analytic for all $T\in\mathbb{R}_+$. Moreover, $H(X)$ can be extended to a function $h(T)$ that is analytic for all $T\in\mathbb{R}_+$ except where $f(T)=0$.
\end{lemma}
\begin{proof}
    We first consider $U(X\mid Y)$ and $U(Y\mid X)$, which are defined in Eqs.~\eqref{eq:pred}--\eqref{eq:recon}. These can be interpreted as functions of $T\in\mathbb{Z}_+$ whose values belong to the interval $[0,1]$. According to Guichard's Theorem \cite[Theorem 5.2.1]{Davis1975} (see also \cite[Theorem 15.13]{Rudin1986}), there exist two functions of $z\in\mathbb{C}$, denoted $f$ and $g$, that are holomorphic in the whole complex plane and whose values at $z=T\in\mathbb{Z}_+$ equal those of $U(X\mid Y)$ and $U(Y\mid X)$, respectively. 
     
    Now, $U(X\mid Y)$ and $U(Y\mid X)$, and consequently $f(z)$ and $g(z)$, have bounded values for all $z=T\in\mathbb{Z}_+$. Moreover, $f$ and $g$ are holomorphic, so their restriction to the axis $z=T\in \mathbb{R}$ is real analytic. Hence, on that axis, $f$ and $g$ are Lipschitz continuous, which means that there are positive and finite constants, $a$ and $b$, such that $|f(T)-f(T')|\leq a|T-T'|$ and $|g(T)-g(T')|\leq b|T-T'|$ for all $T,T'\in \mathbb{R}$. Choosing $T=T'+\epsilon$ with $T'\in\mathbb{Z}_+$ and $|\epsilon| <1 $, we conclude that $f(T)$ and $g(T)$ have finite values for all $T\in \mathbb{R}_+$. 
     
    The functions $f$ and $g$ are thus holomorphic in the whole complex plane and bounded on the positive real axis. This allows to  use a special case of Carlson's Theorem \cite[Theorem 2.8.1]{Andrews1999} according to which holomorphic functions that are bounded on the positive real axis are uniquely defined by their values on the set $\mathbb{Z}_+$. Therefore, $f$ is the unique extension $U(X\mid Y)$ that is holomorphic for all $T\in\mathcal{C}$. Note that the restriction of $f$ on the positive real axis is real analytic on this domain. 
    Thus, there is a unique extension of  $U(X\mid Y)$ that is real analytic for all $T\in \mathbb{R}_+$ and that can be further extended to a holomorphic function for all $T\in \mathbb{C}$. 
    The same conclusion holds for $g$ and $U(Y\mid X)$. 
     
    To finish the proof, we need to tackle $H(X)$. We cannot use the same strategy as above because $H(X)$ is not a bounded function of $T\in \mathbb{Z}_+$. However, by definition, the identity 
    \begin{equation}\label{eq:identity_U}
       H(X)  = \frac{H(Y) U(Y\mid X)}{U(X\mid Y)}\,.
    \end{equation}
    is valid whenever $U(X\mid Y)>0$. Now, according to Lemma  \ref{lem:monotonicity}, $I(X;Y)>0$ and hence $U(X\mid Y)>0$ for all $T\in \mathbb{Z}_+$. This means that Eq.~\eqref{eq:identity_U} is well defined for all $T\in \mathbb{Z}_+$. To extend the domain of validity of the identity, we use the analytic functions $f$ and $g$ introduced above and define a new function $h$ as 
    \begin{equation}
       h(T)  = H(Y) \frac{g(T)}{f(T)}\,.
    \end{equation}
    The values of $h$ coincide with those of $H(X)$ for all $T\in\mathbb{Z}_+$, so that Eq.~\eqref{eq:identity_U} defines a unique extension of $H(X)$. Moreover, $h$ is analytic for all $T\in \mathbb{R}_+$ except at the points $T$ where $f(T)=0$.  
\end{proof}

Lemma~\ref{lem:extension} ensures the \emph{existence of analytic extensions} for the uncertainty coefficients, considered as functions of the positive integer $T$. 
These extensions can thus be evaluated and derived without restriction on the whole domain $\mathbb{R}_+$, which is a desirable property that will soon be exploited. 
However, the same lemma does not guarantee the monotonicity of the extensions on $\mathbb{R}_+$ in the event where they are monotone on $\mathbb{Z}_+$, although we will assume that it is the case from now on. 
This is a reasonable assumption since numerical methods, generalizing the well-known Fritsch-Butland algorithm \cite{fritsch1984method}, have been recently developed to \emph{construct smooth} (i.e., at least continuously differentiable) \emph{and monotone interpolating functions} from any finite monotone datasets \cite{wolberg2002energy, yao2018unconditionally}. 
With this assumption in hand, together with Lemmas~\ref{lem:monotonicity} and \ref{lem:extension}, we now proceed to prove our main theoretical result: the universality of the $T$-duality in Markov chains.

\begin{theorem}\label{thm:universality}
    Let $X=(X_1, X_2,\cdots, X_T)$ and $Y$ respectively be a Markov chain and a discrete random variable as in Lemma \ref{lem:monotonicity}. Additionally, we suppose that $X$ has a finite nonzero entropy rate and that $Y$ has a nonzero entropy. Then there exists a positive constant $\tau$ such that the uncertainty coefficients $U(Y\mid X)$ and $U(X\mid Y)$ are $T$-dual for all $T \geq \tau$.
\end{theorem}
\begin{proof}
    According to Lemma \ref{lem:extension}, the quantities $U(X\mid Y)$, $U(Y\mid X)$, and $H(X)$, which were originally defined as real functions of $T\in\mathbb{Z}_+$, have unique analytic extensions on the positive real axis, i.e., $T\in\mathbb{R}_+$. This allows us to treat $U(X\mid Y)$, $U(Y\mid X)$, and $H(X)$ as continuously differentiable functions with respect to $T$, where $U(Y\mid X) = \frac{I(X;Y)}{H(Y)}$ and $H(X)$ are also monotone.

     Now, by hypothesis, the entropy rate of the Markov chain $X$, $R := \lim_{T\to \infty}\frac{H(X)}{T}$, is well defined and nonzero.
    Hence, $H(X) \sim RT$, i.e., $H(X)$ is positive and asymptotically linearly increasing with $T$.
    Moreover, since $Y$ is independent of $T$ and $I(X;Y)>0$, it follows that $I(X;Y)$ is monotonically increasing with respect to $T$ by Lemma~\ref{lem:monotonicity}. As a result, $U(Y\mid X) = \frac{I(X;Y)}{H(Y)}$ is also monotonically increasing, since its denominator is independent of $T$, by assumption. This translates to the strict inequality $\pder{U(Y\mid X)}{T} > 0$. If there exists a $T$-duality, i.e., there is a domain of $T$ where Eq.~\eqref{eq:duality-criterion} is true, then $U(X\mid Y)$ must be monotonically decreasing with $T$---or $\pder{U(X\mid Y)}{T} < 0$---in that domain. To prove this, note that we can relate the two uncertainty coefficients using Eq.~\eqref{eq:identity_U}.
    This leads to the following differential equation
    \begin{equation}
        \pder{}{T}[\log U(X\mid Y)] = \pder{}{T}[\log U(Y\mid X)] - \pder{}{T}[\log H(X)]\,,
    \end{equation}
    where we used the fact that $\pder{H(Y)}{T} = 0$. Hence, to show that $U(X\mid Y)$ is monotonically decreasing with $T$, the following inequality must hold
    \begin{equation}\label{eq:goal_inequality}
        \pder{}{T}[\log U(Y\mid X)] < \pder{}{T}[\log H(X)] \,.
    \end{equation}

    Suppose for a moment that $U(X\mid Y)$ is in fact increasing, such that Eq.~\eqref{eq:goal_inequality} is false. This will eventually give rise to a contradiction. 
    Let $g(T):=U(Y\mid X)$ and $h(T):= H(X)$ be continuous functions of $T$ such that their derivative with respect to $T$ are respectively given by $g'(\tau) := \left.\pder{f(T)}{T}\right|_{T=\tau}$ and $h'(\tau) := \left.\pder{h(T)}{T}\right|_{T=\tau}$. Note that $0<f(\tau)\leq 1$ and $h(\tau)>0$ for all $\tau\in\mathbb{R}_+$. If Eq.~\eqref{eq:goal_inequality} is false, then
     \begin{equation}
        (\log g(T))' \geq (\log h(T))'\,.
    \end{equation}
    Using Gr\"{o}nwall's inequality \cite[Theorem 1.2.1]{lakshmikantham1969differential}, we get
     \begin{equation}\label{eq:upper_bound_f}
    \frac{g(T)}{g(a)} \geq \frac{ h(T)}{h(a)}\,,\qquad 0 < a < T.
    \end{equation}
    So far, we have established that $h(T)=H(X) \sim RT$ and that $U(Y\mid X)$ is  monotonically increasing. We have also proved that if $U(X\mid Y)$ is not monotonically decreasing with $T$, then inequality~\eqref{eq:upper_bound_f} is satisfied. However, the latter inequality and $h(T) \sim RT$ readily imply that $g(T)$ belongs to the class $\Omega(T)$, which is the set of all $\tilde{g}(T)$ such that there exist positive constants, $S$ and $T^*$, for which $\tilde{g}(T)\geq S T$ for all $T \geq T^*$ 
    (i.e., Knuth's Big Omega~\cite{Knuth1976}). 
    
    Two cases must be considered. First, if $ST^* > 1$, then $\tilde{g}(T)\geq S T^* > 1$, which is in direct contradiction with $g(T) \leq 1$ whenever $T\geq T^*$. Second, if $ST^* \leq 1$, then choose $T^{**} > S^{-1} \geq T^*$, so that  $\tilde{g}(T) \geq  ST^{**} > 1$ for all $T\geq T^{**}$. This again contradicts the inequality $g(T) \leq 1$ whenever $T\geq T^{**}$. As a result, inequality~\eqref{eq:upper_bound_f} cannot be satisfied when $T\geq\tau$, with $\tau = \max\set{T^*, T^{**}}$. We thus conclude that $U(X\mid Y)$ is monotonically decreasing for all $T\geq \tau$. Therefore, $U(Y\mid X)$ and $U(X\mid Y)$ are $T$-dual in the interval $[\tau,\infty)$.
\end{proof}

\subsection{Past-dependent mutual information}
\label{app:past_mi}
\begin{figure*}
    \centering
    \includegraphics{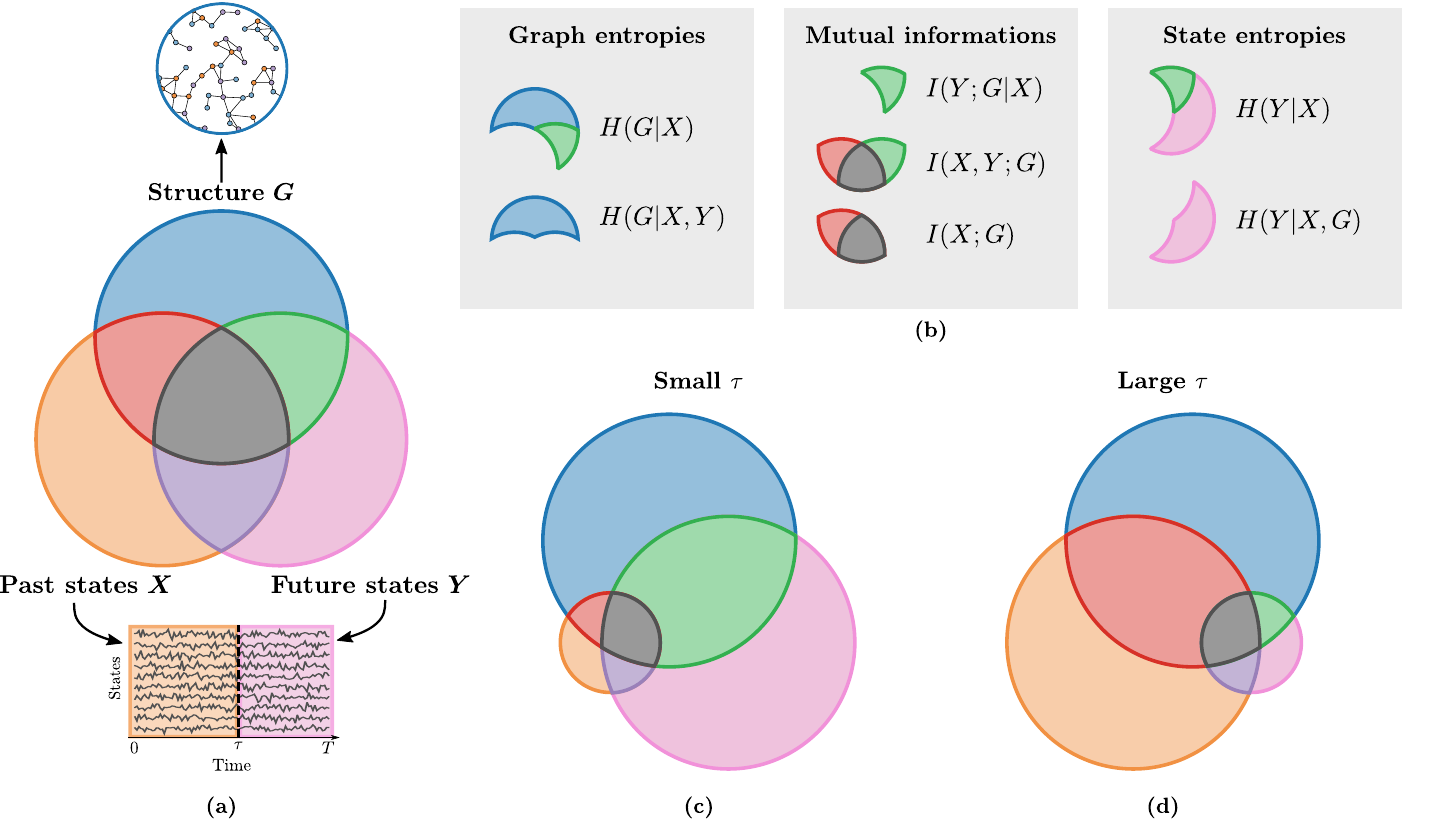}
    \caption{\textbf{Information diagrams for the past-dependent information measures}. On panel (a), we show the information diagram of the random variable triplet $(X, Y, G)$, where $X$ represents the past states, $Y$, the future and $G$, the structure of the system.  On panel (b), we highlight the quantities of interest for computing $I(Y;G \mid X)$, that is graph (left) and state (right) entropies involved in Eq.~\eqref{eq:cond_mi2}, and the mutual informations (middle) in Eq.~\eqref{eq:cond_mi1}. Panels (c) and (d) shows two extreme scenarios where the length of the past $\tau$ is small and large, which illustrates how the different information measures change with $\tau$.}
    \label{fig:figure5}
\end{figure*}
We present a generalization of the mutual information in which the Markov chain, hereafter denoted $Z = (X, Y)$, is partitioned into two parts, namely the past states $X$ and the future states $Y$, both conditioned on a random graph $G$. 
The past $X$ and the future $Y$ are both Markov chains, of respective length $\tau$ and $T-\tau$, where $T$ is the complete length of $Z$. 
By separating the past from the future, we can define new information measures that are closer to more standard predictability measures \cite{delsole2007predictability, giannakis2012information} interested in quantifying how knowledge about the past influences our capacity to predict the future.
In this new scenario, we define the past-dependent mutual information as follows:
\begin{align}\label{eq:cond_mi1}
    I(Y;G \mid X) = I(X,Y; G) - I(X;G)\,,
\end{align}
which is a conditional mutual information, where both $I(X,Y;G)$ and $I(X;G)$ can be expressed as before from Eq.~\eqref{eq:mutual_info}. As illustrated by the information diagram of Fig.~\ref{fig:figure5}(a), we can expand this mutual information in two ways, using either the dynamical or the structural interpretations:
\begin{equation}\label{eq:cond_mi2}
    \begin{split}
        I(Y;G \mid X) &= H(Y \mid X) - H(Y \mid X, G)\\ 
        &= H(G \mid X) - H(G \mid X,Y)\,.
    \end{split}
\end{equation}
These information measures are highlighted by Fig.~\ref{fig:figure5}(b). 
Similarly to Sec.~\ref{sec:information-theory}, we then define the partial uncertainty coefficients, bounded between 0 and 1:
\begin{subequations}\label{eq:partial_uncertainty_coefficients}
\begin{align}
    U_X(Y \mid G) &= \frac{I(Y;G \mid X)}{H(Y \mid X)}\,,\\
    U_X(G \mid Y) &= \frac{I(Y;G \mid X)}{H(G \mid X)}\,,
\end{align}
\end{subequations}
measuring the partial predictability of $Y$ from $G$ and partial reconstructability of $G$ given $Y$, respectively.

The physical interpretation of the conditional mutual information $I(Y;G\mid X)$ is very analogous to that of $I(Z;G)$ presented in Sec.~\ref{sec:information-theory}, but still demands further clarifications.
Indeed, it is still a measure of uncertainty reduction between a dynamics $Y$ and a random graph $G$, but where the mutual information associated to the past states $X$ has already been taken into account, as expressed by Eq.~\eqref{eq:cond_mi1}.
Hence, we expect that $I(Y;G \mid X)$ decreases when the length $\tau$ of the past chain $X$ increases [see Figs.~\ref{fig:figure5}(c-d)].
In terms of the relationship between structure and dynamics, the interpretation of $I(Y;G \mid X)$ is less straightforward.
Indeed, the influence of $G$ over $Y$ can be reduced when $X$ is given because a fraction of the structural information is hidden in $X$.
This can potentially be misleading, since it does not necessarily imply that the influence of $G$ over the complete dynamics has been reduced whatsoever.
The behaviors of $U_X(Y\mid G)$ and $U_X(G\mid Y)$ can still result in $\theta$-dualities, as shown below, but these dualities are harder to interpret because of the structural information hidden in $X$. 
%

\begin{figure*}
    \centering
    \includegraphics{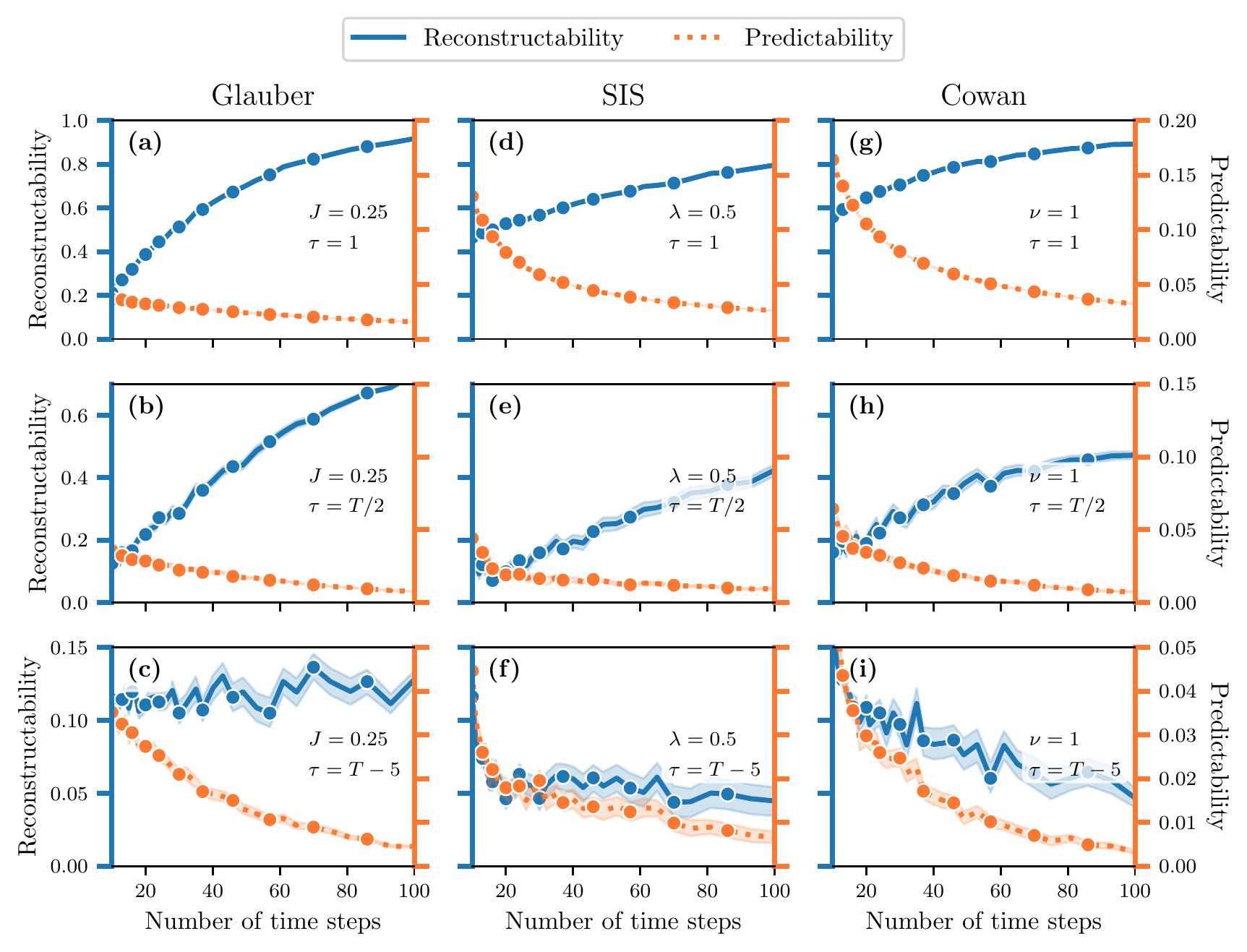}
    \caption{
        \textbf{Existence of the $T$-duality in the past-dependent case, for binary dynamics evolving on small Erd\H{o}s-R\'enyi random graphs}:
        (a-c) Glauber dynamics,
        (d-f) SIS dynamics and
        (g-i) Cowan dynamics.
        Like Fig~\ref{fig:exact-duality-timestep}, each panel shows the reconstructability coefficient $\recon\in[0,1]$ (blue) and the predictability coefficient $\pred\in[0,1]$ (orange) as a function of the number of time steps $T$, and .
        In each row, we change the value of the length $\tau$ of the past Markov chain $X$: (a,d,g) $\tau=1$, (b,e,h) $\tau = T/2$ and (c,f,i) $\tau= T - 5$.
        We used graphs of $N=5$ vertices and $E=5$ edges and each symbol corresponds to the average value measured over $2000$ samples.
        We also show different values of the coupling parameters, as indicated on each figure.
    }
    \label{fig:figure6}
\end{figure*}

Using the partial uncertainty coefficients in Eqs.~\eqref{eq:partial_uncertainty_coefficients}, we investigate the $T$-duality for different values of $\tau$.
Two different scenarios are of interest: The case where $\tau$ is constant with respect to $T$ and the case where it is not. 
When $\tau$ is constant with respect to $T$, Theorem~\ref{thm:universality} remains valid since the additional conditions on the Markov chain $Y$ and the random graph $G$ are a special case of the prior assumptions.
Hence, we observe the $T$-duality for any value of $\tau$ in this case, as supported by Figs.~\ref{fig:figure6}(a,d,g).

The second scenario, when $\tau$ is a function of $T$, is more nuanced, as seen in Figs.~\ref{fig:figure6}(b-c,e-f,h-i) since Theorem~\ref{thm:universality} no longer applies.
This is because both $Y$ and $G$ (represented by $X$ and $Y$ in Theorem~\ref{thm:universality}, respectively) are now conditioned on $X$, and thus will depend on $T$.
Consequently, we no longer can assume that the entropy rate of $Y$ given $X$ is constant with $T$ and that $H(G\mid X)$ is independent of $T$.
In Fig.~\ref{fig:figure6}, we break this scenario into two cases.
We consider $\tau = \kappa T$ [Figs.~\ref{fig:figure6}(b,e,h) with $\kappa=\frac12$], where the lengths of $X$ and $Y$ remain proportional to one another. 
In this case, the $T$-duality seems to persist for all three dynamics.
However, when $\tau = T - \xi$ [Figs.~\ref{fig:figure6}(c,f,i) with $\xi=5$] where the size of $Y$ remains fixed and $X$ grows linearly with $T$, the $T$-duality is no longer observed, except for the Glauber dynamics.
It is important to note that, for small $\xi$, the partial reconstructability coefficient $U_X(G\mid Y)$ becomes numerically unstable since both $I(Y;G\mid X)$ and $H(G\mid X)$ tend to zero.
This is why the curves are much noisier in that case.
Informed by these examples, we make the following conjecture:
\begin{conjecture}\label{conj:tau_scaling}
    Let $Z=(X,Y)$ be a Markov chain, composed of the two consecutive Markov chains $X$ and $Y$ of respective length $\tau$ and $T-\tau$, such that $Z$ is conditioned on a discrete random variable $W$. Then, there exists a function $g(T)$ such that, if $\tau$ is dominated by $g(T)$, the partial uncertainty coefficients $U_X(Y\mid W)$ and $U_X(W\mid Y)$ are $T$-dual, and they are not otherwise.
\end{conjecture}

\subsection{Estimators of the mutual information}
\label{app:estimators}
The mutual information $\mi$ is generally intractable. Its intractability stems from the evaluation of the \emph{evidence} probability, which is defined by the following equation:
\begin{equation}\label{eq:evidence}
    P(X) = \sum_{G\in\mathcal{G}_N} P(G) P(X\mid G)\,.
\end{equation}
Indeed, this sum potentially counts a number of terms which grows exponentially with the number of vertices $N$ in the random graph. More specifically, the evidence probability appears in two entropy terms needed to compute the mutual information, namely the marginal entropy $H(X) = -\mean{\log P(X)}$ and the reconstruction entropy $H(G \mid X) = -\mean{\log \frac{P(G)P(X \mid G)}{P(X)}}$, where $\mean{f(Y)}$ denotes the expectation of $f(Y)$. Fortunately, the evidence probability, and in turn the mutual information, can be estimated efficiently using Monte Carlo techniques, which we present in this section.

\subsubsection{Graph enumeration approach}
For sufficiently small random graphs ($N\leq 5$), the evidence probability can be efficiently computed by enumerating all graphs of $\mathcal{G}_N$ and by adding explicitly each term of Eq.~\eqref{eq:evidence}. Then, we can estimate the mutual information by sampling $M$ graph-states pairs, denoted $(G^{*(m)}, X^{*(m)})$, and by computing the following arithmetic average:
\begin{equation}
    \begin{split}
        I(X;G) \simeq \frac{1}{M} \sum_{m=1}^M &\log P\Big(X^{*(m)} \mid G^{*(m)}\Big) \\&
        - \log P\Big(X^{*(m)}\Big)\,.
    \end{split}
\end{equation}
The variance of this estimator scales with the inverse of $\sqrt{M}$. In Fig.~\ref{fig:exact-duality-timestep}, we used this estimator to compute the mutual information, where $M = 1000$.

\subsubsection{Variational mean-field approximation}
In this approach, we estimate the posterior probability instead of the evidence probability. According to Bayes' theorem, the posterior probability is 
\begin{equation}
    P(G \mid X) = \frac{P(G)P(X \mid G)}{P(X)}\,.
\end{equation}
Behind this estimator is a variational mean-field (MF) approximation that assumes the conditional independence of the edges. For simple graphs, the MF posterior is
\begin{equation}\label{sieq:mf_posterior}
    P_\mathrm{MF}(G \mid X) = \prod_{i\leq j} [\pi_{ij}(X)]^{A_{ij}}\, [1 - \pi_{ij}(X)]^{1 - A_{ij}}\,,
\end{equation}
where $\pi_{ij}(X):= P(A_{ij}=1 \mid X)$ is the marginal conditional probability of existence of the edge $(i, j)$ given $X$. For multigraphs, a similar expression can be obtained, but instead involves a probability $\pi_{ij}(\omega\mid X) := P(M_{ij}=\omega|X)$ that there are $\omega$ multi-edges between $i$ and $j$. In this case, the MF posterior becomes
\begin{equation}
    P_\mathrm{MF}(G \mid X) = \prod_{i<j}\prod_{\omega=0}^\infty [\pi_{ij}(\omega\mid X)]^{\delta_{\omega,M_{ij}}}\,,
\end{equation}
where $\delta_{x,y}$ is the Kronecker delta. The MF approximation allows to compute a lower bound of the true posterior entropy, such that
\begin{equation}\label{sieq:mf_lowerbound}
    H(G \mid X) \geq -\mean{\log P_\mathrm{MF}(G \mid X)}\,,
\end{equation}
as a consequence of the conditional independent between the edges \cite[Theorem 2.6.5]{cover2006elements}. Using the MF approximation and a strategy similar to the exact estimator, we compute the MF estimator of the mutual information as follows:
\begin{equation}\label{sieq:mf_estimator}
    \begin{split}
        I(G \mid X) \geq \frac{1}{M}\sum_{m=1}^M \Big[&\log P_\mathrm{MF}\Big(G^{*(m)} \mid X^{*(m)}\Big) \\&- \log P\Big(G^{*(m)}\Big)\Big]\,.
    \end{split}
\end{equation}
To compute $P_\mathrm{MF}\Big(G^{*(m)} \mid X^{*(m)}\Big)$, we sample a set $\mathcal{Q}^{(m)} := \set{G_1^{*(m)},\cdots, G_Q^{*(m)}}$ of $Q$ graphs from the posterior distribution $P(G\mid X^{*(m)})$.
Then, we estimate the probabilities $\pi_{ij}(X) \simeq \frac{n_{ij}^{(m)}}{Q}$ using their corresponding maximum likelihood estimate, where $n_{ij}^{(m)}$ is the number of times the edge $(i, j)$ is seen in $\mathcal{Q}^{(m)}$.
An analogous maximum likelihood estimate is made in the multigraph case, where $\pi_{ij}(\omega\mid X) \simeq \frac{n_{ij;\omega}^{(m)}}{K}$ and $n_{ij;\omega}^{(m)}$ counts the number of times there were $\omega$ multiedges between $i$ and $j$ in $\mathcal{Q}^{(m)}$. This estimator is a lower bound of the mutual information---a consequence of Eq.~\eqref{sieq:mf_lowerbound}. Hence, it is biased, and the extent of this bias is dependent on the quality of the conditional independence assumption with respect to the true random graph. Note that the MF estimator can yield negative estimates of the mutual information (see Fig.~\ref{fig:numerical-validation}).

In Figs.~\ref{fig:numerical-validation} and \ref{fig:approx-duality-coupling}, we fix the number of graphs sampled from the posterior distribution to $Q=1000$, and propose $5N$ moves between each sample, as also mentioned in App.~\ref{app:mcmc-algo}.

\subsubsection{Annealed important sampling}

Whereas the MF estimator represents a biased estimator of the posterior probability $P(G \mid X)$, there exists other Markov chain Monte-Carlo (MCMC) techniques that tackle the problem of estimating the evidence probability directly. The one we consider in this paper is obtained from an \emph{annealed importance sampling} (AIS) procedure called the stepping-stone (SS) algorithm~\cite{xie2011improving}.

The procedure of the stepping-stone algorithm takes advantage of the fact that it is possible to sample efficiently from the posterior distribution $P(G \mid X)$ using MCMC (see Section~\ref{app:mcmc-algo}). In order to compute an accurate estimator of the evidence probability $P(X)$, the procedure samples the space $\mathcal{G}_N$ according to $P_\beta(G \mid X)$, where $0 \leq\beta\leq 1$ is an inverse temperature parameter that dampens the influence of the likelihood such that
\begin{equation}
    P_\beta(G \mid X) \propto [P(X \mid G)]^\beta P(G)\,.
\end{equation}
The inverse temperature basically allows the Markov chain to navigate $\mathcal{G}_N$ efficiently to construct an accurate estimator of $P(X)$, that is where the graph samples are not all too close or too far from the maximum posterior. More specifically, the AIS estimator is defined by
\begin{equation}
    P_\mathrm{AIS}(X) = \prod_{k=1}^K \mean{[P(X \mid G_k^*)]^{\beta_{k} - \beta_{k-1}}}\,,
\end{equation}
where $0= \beta_0 < \cdots < \beta_{K - 1} = 1$ and the expectation is evaluated with respect to $G_k^* \sim P_{\beta_k}(G \mid X^*)$, for each $k$. Similarly to the mean-field estimator, we estimate this expectation by collect a sample $\mathcal{Q}^{(m)}_k$ of $Q$ graphs distributed according to $P_{\beta_k}(G \mid X^{*(m)})$, for each $k$.

Taking the log of this equation gives us an estimator of the log-evidence probability, which we can use to compute the mutual information directly:
\begin{equation}
    \log P_\mathrm{AIS}(X) = \sum_{k=1}^K \log \left\{\mean{[P(X \mid G_k^*)]^{\beta_{k} - \beta_{k-1}}}\right\}\,.
\end{equation}
Although the estimator for $P_\mathrm{AIS}$ is unbiased, the one for the log-evidence probability introduces a bias:
\begin{equation}
    \log P(X) \geq \log P_\mathrm{AIS}(X)\,.
\end{equation}
This bias can be arbitrarily reduced by increasing $K$~\cite{xie2011improving}, although we found that doing so provides diminishing returns.
Using the AIS estimator of the evidence probability, we obtain an AIS estimator of the mutual information such that
\begin{equation}
    \begin{split}
        I(G;X) \leq \frac{1}{M}\sum_{m=1}^M \Big[&\log P\Big(X^{*(m)} \mid G^{*(m)}\Big) \\& -\log P_\mathrm{AIS}\Big(X^{*(m)}\Big)\Big]\,.
    \end{split}
\end{equation}

Following Ref.~\cite{xie2011improving}, we use values of $\beta_k$ distributed according to a beta distribution $\mathrm{Beta}(\alpha, 1)$, where $\beta_k = \left(\frac{k}{K}\right)^{1/\alpha}$, such that increasing $\alpha$ controls how skewed around zero the sequence $\set{\beta_k}_k$ is. For Fig.~\ref{fig:numerical-validation}, we fix $\alpha=0.5$ and $K=20$ and, for each value of $\beta_k$, we sample $1000$ graphs from $P_{\beta_k}(G \mid X^*)$, proposing $5N$ moves in-between each sample (see Appendix~\ref{app:mcmc-algo}).

\subsubsection{Evaluation of the mutual information in large systems}

Next, we evaluate the quality of each estimator on small and large systems. Figure~\ref{fig:numerical-validation}(a) shows the behavior of $\mi$ in the Glauber dynamics on a small Erd\H{o}s-R\'enyi random graph as approximated using the MF and AIS estimators, and compares them to an exact evaluation based on an explicit graph enumeration used in Fig.~\ref{fig:exact-duality-timestep}.
As expected the two estimators provide a lower and an upper bound for $\mi$, and these bounds are fairly tight.

Several caveats are in order.  On the one hand, the bias of the AIS estimator can, in principle, be reduced arbitrarily by increasing the number $K$ of temperature steps, but its evaluation becomes quickly computationally costly.  On the other hand, the evaluation of MF estimator is comparatively quicker, but cannot be improved by further sampling.
The AIS estimator is accordingly closer to the exact value throughout, but it can sometimes overestimate the mutual information above its upper bound since $H(X)$ is overestimated while $H(X\mid G)$ is not.
The MF estimator can also yield negative values of $\mi$ for small values of $J$---i.e., regimes where $H(G\mid X) \simeq H(G)$---due to an overestimated $H(G\mid X)$ becoming larger than $H(G)$.

Figure~\ref{fig:numerical-validation}(b) shows the same experiment as in Fig.~\ref{fig:numerical-validation}(a) but with larger graphs of $N=100$ vertices and leads to similar observations: the AIS estimator is always greater than the MF estimator, and both estimators sometimes yields approximated values for $\mi$ outside of the valid range [0,$\max\set{H(G), H(X)}$].
Interestingly, these bounds are nevertheless fairly close to one another, as in the case $N=5$.


\begin{figure}
    \centering
    \includegraphics[scale=1]{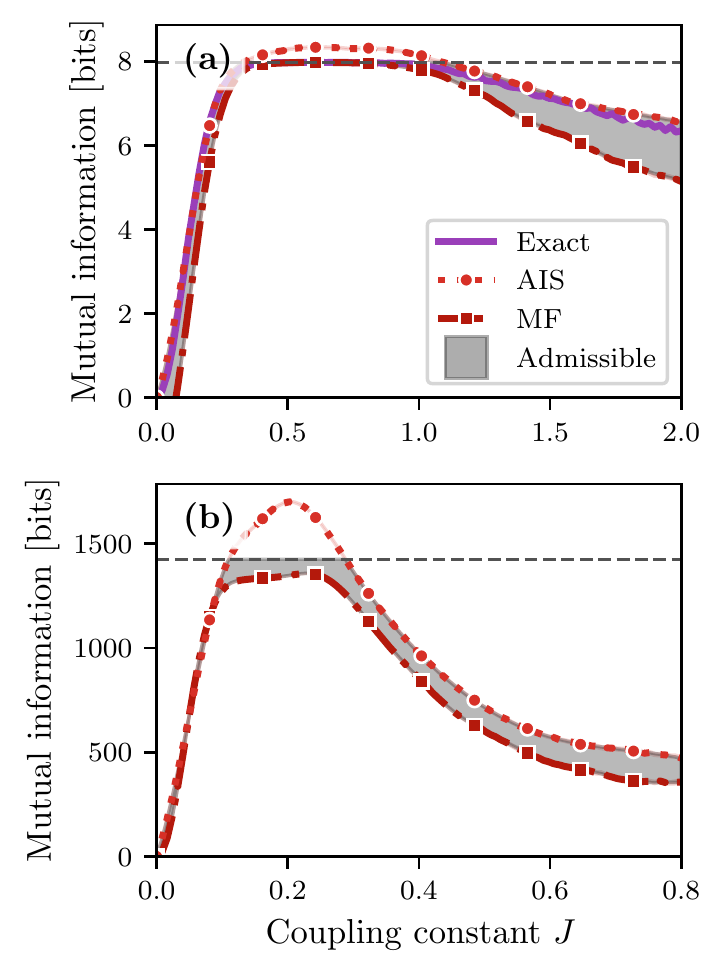}
    \caption{\textbf{Estimators of the mutual information in the Glauber dynamics on Erd\H{o}s-R\'enyi graphs as a function of the normalized coupling parameter $J\mean{k}$}: (a) $N=5$, $E=5$ and $T=100$ (b) $N=100$, $E=250$ and $T=1000$. The solid line in (a) corresponds to the exact evaluation of $\mi$ and is the same line as the one in Fig.~\ref{fig:exact-duality-timestep}(a). The circles and square in both (a) and (b) represent the values of $\mi$ computed using the AIS and the MF estimators, respectively. The dashed line indicates the upper bound of $\mi$, i.e., $\max\set{H(G), H(X)}$. We also show with a gray area the admissible values of $\mi$ bounded by the biased MF and AIS estimators.
    }
    \label{fig:numerical-validation}
\end{figure}

\subsubsection{Biases of the uncertainty coefficients}
When an estimation of the mutual information is biased, it necessarily follows that an estimation of the resulting uncertainty coefficients will also be biased. Fortunately, we can show that the direction of the bias does not change either for the reconstructability $U(G\mid X)$ or the predictability $U(X\mid G)$. Suppose that $\mathcal{I}_\varepsilon = I(X;G)(1 + \varepsilon)$ is an estimator of the mutual information, where $\varepsilon\in\mathbb{R}$ is a small bias which can be either positive or negative. Then, the corresponding estimators of the uncertainty coefficients, that we denote $\mathcal{P}_\varepsilon$ and $\mathcal{R}_\varepsilon$ for the predictability and the reconstructability, respectively, are
\begin{equation}
    \mathcal{P}_\varepsilon = \frac{\mathcal{I}_\varepsilon}{H(X\mid G) + \mathcal{I}_\varepsilon}\,.
\end{equation}
and
\begin{equation}
    \mathcal{R}_\varepsilon = \frac{\mathcal{I}_\varepsilon}{H(G)} = U(G\mid X)(1 + \varepsilon)\,,
\end{equation}
Note that we also suppose that $H(G)$ and $H(X\mid G)$ are not affected by the bias $\varepsilon$. For the first expression, we consider the first-order development of $\mathcal{P}_\varepsilon$ with respect to $\varepsilon$:
\begin{equation}
    \mathcal{P}_\varepsilon = U(X\mid G)\left[1 + \Big(1 - U(X\mid G)\Big)\varepsilon - \bigo{\varepsilon^2}\right]\,.
\end{equation}
Indeed, given that $U(X\mid G)\geq0$, the leading biased term $\Big(1 - U(X\mid G)\Big)\varepsilon$ must have the same sign as $\varepsilon$.
The second expression clearly shows that the bias of $\mathcal{R}_\varepsilon$ is exactly given by $\varepsilon$. Therefore, both $\mathcal{P}_\varepsilon$ and $\mathcal{R}_\varepsilon$ retain the direction of bias of $\mathcal{I}_\varepsilon$.

\subsection{Markov chain Monte-Carlo algorithm}
\label{app:mcmc-algo}

\begin{figure*}
    \centering
    \includegraphics[scale=1]{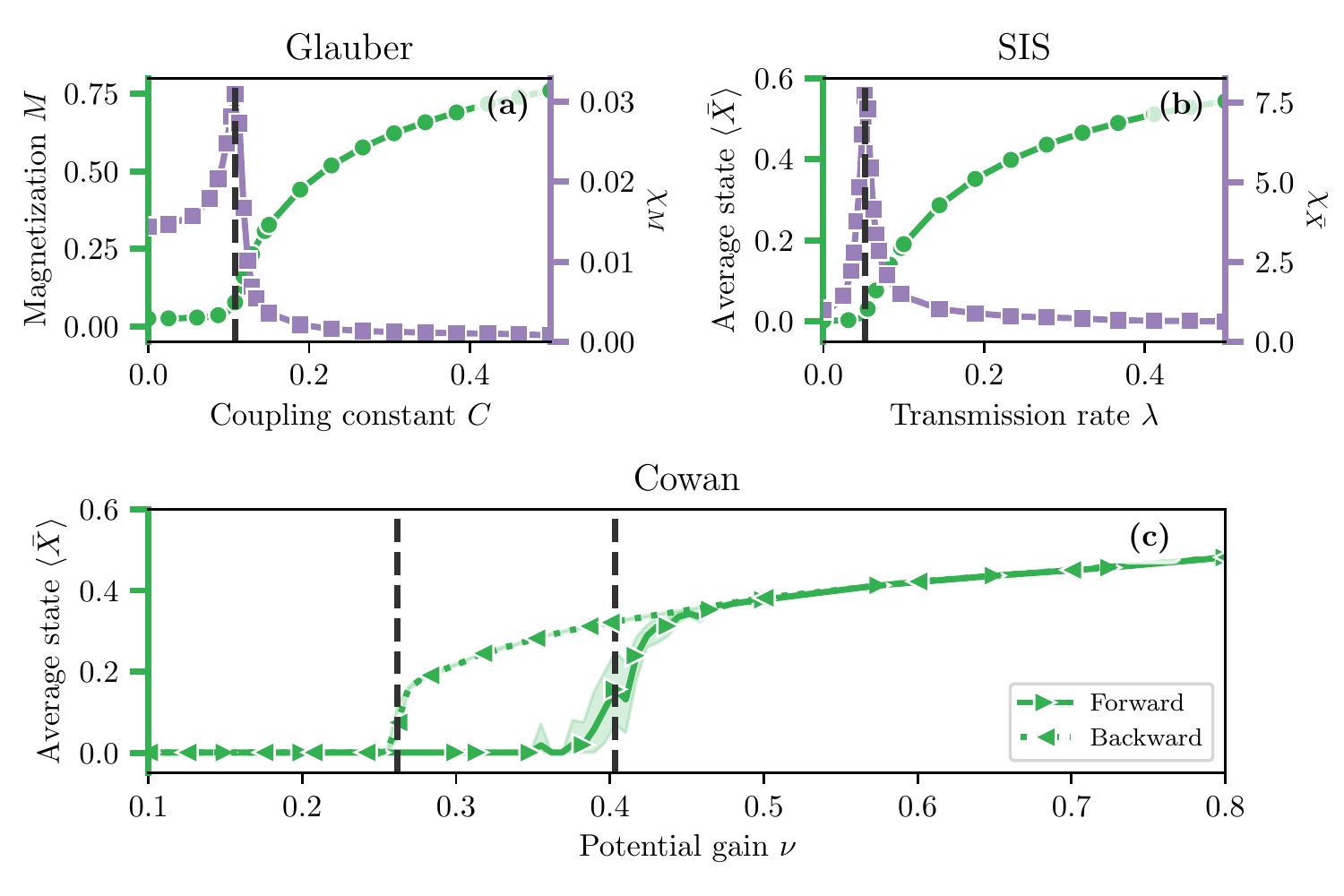}
    \caption{Numerical evaluation of the phase transition thresholds: (a) Glauber dynamics, (b) SIS dynamics, (c) Cowan dynamics. For panels (a) and (b), the left axis (green) shows the order parameter (green circles), and the right axis (purple) shows the susceptibility (purple squares). For panel (c), only the order parameter is shown but for both the forward (right triangle) and backward (left triangle) branches. The values of the thresholds are indicated by the vertical dashed lines. We used the same parameters as those of Fig.~\ref{fig:approx-duality-coupling}, but increased the number of steps $T=10^4$ to better sample from the dynamics. Each marker has been average over 48 realizations.}
    \label{fig:figure7}
\end{figure*}

To sample from the posterior distribution, we use a Markov chain Monte-Carlo (MCMC) algorithm where, starting from a graph $G$, we propose a move, denoted $\bar{G}^{*}\leftarrow G^*$, according to a proposition probability $P(\bar{G}^{*}| G^*)$, and accept it with the Metropolis-Hastings probability:
\begin{equation}\label{eq:accept-prob}
    \min\left(1,  e^{-\log\Delta}\frac{P(G^*|\bar{G}^{*})}{P(\bar{G}^{*}|G^*)}\right)\,,
\end{equation}
where $\Delta = \frac{P(\bar{G}^{*})P(X^*|\bar{G}^{*})}{P(G^*)P(X^*|G^*)}$ is the ratio between the joint probability of the two graphs with $X^*$. This ratio can be computed efficiently in $\bigo{T}$, by keeping in memory $n_{i,t}$, the number of inactive neighbors, and $m_{i,t}$, a number of active neighbors, for each vertex $i$ at each time $t$ (see Ref.~\cite{peixoto2019network}).
Equation~\eqref{eq:accept-prob} allows to sample from the posterior distribution $P(G\mid X)$ without the requirement to compute the intractable normalization constant $P(X)$.
We collect graph samples at every $N\delta$ moves, where we fix $\delta=5$ in all experiments.

We consider two types of random graphs with different constraints: The Erd\H{o}s-R\'enyi model and the configuration model. Hence, we need two different sampling propositions to apply our MCMC algorithm, that is one for each model. We assume that the support of the Erd\H{o}s-R\'enyi model is the set of all simple graphs of $N$ vertices with $E$ edges. In this case, we consider a \emph{hinge flip} move, where an edge $(i, j)$ is sampled uniformly from the edge set of the graph $G$ and a vertex $k$ is sampled uniformly from its vertex set. Then, with probability $\frac{1}{2}$, we rewire edge $(i,j)$ by either selecting $i$ or $j$ to connect with $k$. Note that, because we consider the support $\mathcal{G}_N$ of $G$ to be a space of simple graphs, all moves resulting in the addition of a self-loop or a multiedges are rejected with probability 1. As a result, the proposition probability is the same for any move $\bar{G}^*\leftarrow G^*$:
\begin{equation}
    P(\bar{G}^{*}|G^*) = \frac{1}{E N} \quad\Rightarrow\quad \frac{P(G^*|\bar{G}^*)}{P(\bar{G}^{*}|G^*)} = 1 \,.
\end{equation}
For the configuration model, we assume that the support is the set of all loopy multigraphs of $N$ vertices whose degree sequence is $\vec{k}$. In this case, we propose \emph{double-edge swap} moves according to the prescription of Ref.~\cite{fosdick2018configuring}. We refer to it for further details.

\subsection{Numerical estimation of the phase transition thresholds}
\label{app:thresholds}

We evaluate the phase transition thresholds of each dynamics using standard finite-size scaling techniques and Monte Carlo simulations (see Fig.~\ref{fig:figure7}). For Glauber, an adequate order parameter to visualize the phase transition is the magnetization $M := \frac{1}{N}\sum_i |2X_{i} - 1|$, where the absolute value breaks the spin symmetry~\cite{binder2010monte}. In this process, it is well known that the susceptibility of the order parameter $M$, given by
\begin{equation}
    \chi_M = \frac{\mean{M^2} - \mean{M}^2}{\mean{M}}\,,
\end{equation}
diverges at the threshold $J=J_c$ of the phase transition for infinite size systems~\cite{binder2010monte}. In finite systems, $\chi_M$ instead reaches a maximum at $J=J_c$. We use this fact to locate $J_c$ and show the corresponding results in Fig.~\ref{fig:figure7}(a).

For the SIS dynamics, a similar finite-size scaling analysis can be carried out, but a suitable order parameter is rather the average state $\bar{X} := \frac{1}{N}\sum_{i}X_{i}$. We also use a definition of the susceptibility that is more convenient for spreading processes~\cite{ferreira2012epidemic}, given in terms of $\bar{X}$:
\begin{equation}
    \chi_{\bar{X}} = \frac{\mean{\bar{X}^2} - \mean{\bar{X}}^2}{\mean{\bar{X}}}\,,
\end{equation}
which also diverges at the phase transition threshold $\lambda = \lambda_c$ for infinite size systems. We show the results for SIS in Fig.~\ref{fig:figure7}(b).

Finally, for the Cowan dynamics, we have a first-order phase transition characterized by a discontinuity of the order parameter $\bar{X}$ in the infinite size limit, and a bistable region bounded by two thresholds $\nu_c^b < \nu_c^f$. To find these two thresholds, we evaluate the order parameter $\bar{X}$ for varying values of the parameter $\nu$, and find the location where the discontinuity occurs. We obtain the forward and backward branches by using different initial conditions, where the system is nearly inactive---with one active vertex---and completely active---with no inactive vertex---, respectively.

For the Cowan dynamics, it is important to mention that since we consider relatively small systems ($N=1000$ vertices), the bistable region is not clearly defined. Hence, a system starting in the forward branch can jump on the backward branch with a non-zero probability. This is why the expected discontinuity at the threshold is, in fact, populated (see Fig.~\ref{fig:figure7}(c)). This finite-size effect should be reduced for considering larger systems, but increasing $N$ is unfortunately too computationally costly at the moment. Hence, to get a reasonable estimation of the thresholds in this scenario, we uniformly sample the set of $\nu$'s, compute $\mean{\bar{X}}$ for all values of $\nu$ and find the point $\nu^*$ corresponding to the maximum gap between two points. Then, to increase the precision or this estimation, we zoom on a region centered at $\nu^*$ and do it again, until it converges. This method provides reasonably accurate thresholds for our purposes.

\end{document}